%% file: fossacs-23.tex
\documentclass[runningheads]{llncs}
\usepackage[T1]{fontenc}

\usepackage{orcidlink}

\usepackage{amsmath}
\usepackage{amsfonts}
\usepackage{tikz-cd}
\usepackage{bussproofs}
\usepackage{stmaryrd}
\usepackage{mathabx}
\usepackage{todonotes}
\usepackage{upgreek}
\usepackage{mathtools} 
\usepackage{cmll}  
\usepackage{amssymb}
\usepackage{quiver}

\usepackage{caption}
\usepackage{subcaption}

\usepackage{thm-restate}

\usepackage{hyperref}
\usepackage{cleveref}

\input{pstext}
\input{pscategorytheory}

\input{pssyntax}
\input{pssemantics}

\input{psclonescategoriescl}

\usepackage{color}

\begin{document}
\input{title-and-abstract}

\section{Introduction}
\label{sec:introduction}
\input{sec-introduction}

\section{Multicategories and clones}
\label{sec:clones-and-multicats}
\input{sec-multicategories-and-clones}

\section{Universal properties for multicategories}
\label{sec:universal-properties-for-multicategories}
\input{sec-universal-properties-for-multicats}

%

\section{Product structure}
\label{sec:cartesian-clones}
\input{sec-cartesian-clones}

\subsection{Cartesian structure from representability}
\label{sec:representable-clones}
\input{sec-cartesian-from-representability}

\subsection
	{Recovering the semantic interpretation and syntactic model}
\label{sec:recovering-the-semantics-for-cartesian}
\input{sec-soundness-and-completeness}

\section{Closed structure}
\label{sec:closed-clones}
\input{sec-closed-clones}

\section{Cartesian closed structure}
\label{sec:cc-clones}
\input{sec-cartesian-closed-clones}

\section{Cartesian combinatory logic and SK-clones}
\label{sec:SK-clones}
\input{sec-combinatory-logic-and-SK-clones}

\section{A categorical model of $\stlc$}
\label{sec:SK-categories}
\input{sec-SK-categories}

%
%

\subsubsection{Acknowledgements.} 
I thank Nathanael Arkor and Dylan McDermott for useful discussions on early drafts of this paper, and the reviewers for their many useful comments.
I am grateful to Nayan Rajesh for pointing out the adjunctions between cartesian categories and cartesian clones, and between cartesian closed categories and cartesian closed clones, are in fact equivalences.
Finally, I thank Marcelo Fiore for introducing me to clones.

%

 \bibliographystyle{splncs04}
 \bibliography{clones-combinatory-logic-categories}
%


\vfill
\appendix 

\pagebreak
\section{Omitted proofs}
\label{sec:appendix}
\input{sec-appendix}

\end{document}

%% file: pstext.tex
\newcommand\newmvar[3][]{%
  \newcommand#2[1][1]{#1{\ifcase ##1\or #3\or\else@ctrerr\fi}}%
}

\newcommand{\eg}{{\it e.g.}}

\newcommand{\cf}{{\it cf.}}

\newcommand{\etal}{{\it et al.}}

\newcommand{\suchthat}{\:\big| \: }

\newcommand{\xra}[1]{\xrightarrow{#1}}

\renewcommand{\epsilon}{\varepsilon}  

\newmvar{\param}{P \or Q}

\newmvar{\mapf}{f \or g \or h}
\newmvar{\maph}{h \or k \or m}
\newmvar{\mapp}{p \or q \or r}
\newmvar{\mapu}{u \or v}

\newmvar{\vara}{a \or b}
\newmvar{\varx}{x \or y \or z}
\newmvar{\vari}{i \or j \or k}

\newmvar{\obj}{C \or D \or E } 
\newmvar{\var}{ c \or d \or e } 
\newmvar{\objA}{A \or B \or C \or D}
\newmvar{\objI}{I \or J \or K}
\newmvar{\objP}{P \or Q}
\newmvar{\objR}{R \or S}
\newmvar{\objX}{X \or Y \or Z}

\newmvar{\relR}{R \or S \or K}

\newmvar{\functor}{F \or G \or H}

%% file: pscategorytheory.tex
\newcommand{\smallcatFormat}[1]{\mathcal{#1}}
\newcommand{\cloneFormat}[1]{\mathbb{#1}}

\newmvar{\cat}{\smallcatFormat{C} \or \smallcatFormat{D} \or \smallcatFormat{E}}
\newmvar{\clone}{\cloneFormat{C} \or \cloneFormat{D} \or \cloneFormat{E}}
\newmvar{\multicat}
		{\cloneFormat{M} \or \cloneFormat{N} \or \cloneFormat{O}}

\newcommand{\Nat}{\mathbb{N}}

\newmvar{\indexSet}{\mathbb{I} \or \mathbb{J} \or \mathbb{K}}

\newcommand{\Cat}{\mathbf{Cat}}
\newcommand{\Set}{\mathbf{Set}}

\newcommand{\Clone}{\mathbf{Clone}}
\newcommand{\CartClone}{\mathbf{CartClone}}
\newcommand{\ClClone}{\mathbf{ClClone}}
\newcommand{\StClClone}{  \mathbf{ClClone}_{\mathrm{st}}  }
\newcommand{\CCClone}{\mathbf{CCClone}}

\newcommand{\SKClone}{\mathbf{SKClone}}
\newcommand{\ExtSKClone}{ \mathbf{SKClone}_{\mathrm{ext}} }

\newcommand{\SKCat}{\mathbf{SKCat}}
\newcommand{\StSKCat}{  \mathbf{SKCat} }

\newcommand{\CartCat}{\mathbf{CartCat}}
\newcommand{\CCCat}{\mathbf{CCCat}}
\newcommand{\MonCat}{\mathbf{MonCat}}
\newcommand{\SymMonCat}{\mathbf{SMonCat}}

\newcommand{\Multicat}{\mathbf{Multicat}}
\newcommand{\SymMulticat}{\mathbf{SMulticat}}
\newcommand{\RepMulticat}{\mathbf{RepMulticat}}
\newcommand{\SymRepMulticat}{\mathbf{SRepMulticat}}
\newcommand{\fpMulticat}{\mathbf{fpMulticat}}
\newcommand{\fpSymMulticat}{\mathbf{fpSMulticat}}

\newcommand{\ClMulticat}{\mathbf{ClMulticat}}
\newcommand{\ClSymMulticat}{\mathbf{ClSMulticat}}

\newcommand{\Sig}{\mathbf{Sig}}
\newcommand{\sigFunctor}{f}

\newcommand{\op}[1]{  {#1}^{\mathrm{op}} }

\newcommand{\monoidal}[1]{ (#1, \tens, \tensu) }
\newcommand{\monoidalclosed}[1]{ (#1, \tens, \tensu, \lolliObj{-}{=}) }
\newmvar{\tensu}{I \or J}
\newcommand{\tens}{\otimes}
\newmvar{\monPhi}{\phi \or \alpha}
\newmvar{\monPsi}{\psi \or \beta}

\newcommand{\tinyprod}{ { \mathrm{\Pi} } }

\newcommand{\terminal}{1}
\newcommand{\cartesian}[1]{ (#1, \tinyprod) }
\newcommand{\ccc}[1]{ (#1, \tinyprod, \To) }

\newcommand{\act}{\bullet}

\newmvar{\car}{A \or B \or C} 
\newmvar{\alg}{a \or b \or c} 

\newmvar{\open}{U \or V}

\newcommand{\forget}{ \mathrm{U} } 
\newcommand{\free}{  \mathrm{F} } 

\newmvar{\nattrans}{\nu \or \xi \or \zeta \or \kappa}
\newmvar{\twocell}{\nu \or \xi \or \zeta \or \kappa}

\newcommand{\nameOf}[1]{  \ulcorner{#1}\urcorner  }

\newcommand{\Obj}[1]{ |#1| }
\newcommand{\Id}{\mathrm{Id}}
\newcommand{\id}{\mathrm{id}}

\newcommand{\eval}{\mathrm{eval}}

\newmvar{\cone}{\varpi \or \xi \or \zeta}
\newcommand{\ext}[1]{{#1}^{\#}}

\newmvar{\medmap}{u \or v}
\newmvar{\prodPres}{\nu \or \varpi}
\newcommand{\inc}{\iota}
\newcommand{\seq}[1]{\langle #1 \rangle}
\newcommand{\seqlr}[1]{\left\langle #1 \right\rangle}
\newcommand{\seqbig}[1]{\big\langle #1 \big\rangle}

\newcommand{\adjUp}{{\mathbin{\rotatebox[origin=c]{270}{$\scriptstyle\dashv$}}}}
\newcommand{\smalladjUp}{{\mathbin{\rotatebox[origin=c]{270}{$\scriptstyle\dashv$}}}}
\newcommand{\To}{\Rightarrow}

\newcommand{\iso}{\cong}
\newcommand{\smallprod}{\textstyle{\prod}}
\newmvar{\globalelt}{g \or g'} 

\newcommand{\expObj}[2]{ {{#1} \To {#2}} }
\newcommand{\lolliObj}[2]{ {[#1, #2]} }

\newcommand{\lolliObjbig}[2]{ {\big[#1, #2 \big]} }

\newmvar{\fib}{p \or q} 
\newmvar{\total}{\smallcat E \or \smallcat E'}
\newmvar{\base}{\smallcat B \or \smallcat B'}

\newcommand{\sorts}{S}
\newmvar{\mmap}{t \or u \or v}

\newmvar{\cloneFunctor}{f \or g \or h}
\newmvar{\multicatFunctor}{f \or g \or h}

\newcommand{\nucleus}[1]{  \overline{#1}  }

\newcommand{\univ}{\rho} 

\newcommand{\toMulti}{\mathrm{M}}
\newcommand{\catToClone}{\mathcal{P}}
\newcommand{\monToMulti}{\mathcal{T}}
\newcommand{\multiToMon}{\mathcal{N}}

\newcommand{\msub}[2]{ {#1} \circ \seq{#2} }

\newcommand{\csublr}[2]{ {#1}{ \left[ {#2} \right] } }
\newcommand{\csubbig}[2]{ {#1}{ \big[ {#2} \big] } }
\newcommand{\csub}[2]{ {#1}[{#2}] }
\newcommand{\csubthree}[3]{ \csub{ \csub{#1}{#2} }{ #3 } }

\newcommand{\p}[2]{  \mathsf{p}_{#1}^{#2} }

\newcommand{\closed}[1]{ (#1, \To, \eval) }

\tikzcdset{scale cd/.style={every label/.append style={scale=#1},
    cells={nodes={scale=#1}}}}
\tikzcdset{scalenodes/.style=
  {cells={nodes={scale=#1}}}}

\newenvironment{td}[0]{\begin{center}\vspace{-0.0em}\begin{tikzcd}}{\end{tikzcd}\end{center}}

%% file: pssyntax.tex
\newcommand{\sig}{\mathcal{S}}

\newcommand{\binddot}{\, . \,}
\newcommand{\wkn}[2]{ {#1}^{#2}}

\newcommand{\ind}[1]{  {#1}_{\bullet}  }

\newcommand{\lollipop}{\multimap}

\newcommand{\basicLambda}{ \Lambda }
\newcommand{\stpc}{\basicLambda^{\times}}
\newcommand{\stlc}{\basicLambda^{\to}}
\newcommand{\stlpc}{\basicLambda^{\times, \to}}

\newcommand{\basicLin}{ \mathfrak{L} }
\newcommand{\linTens}{ \basicLin^{\tens} }
\newcommand{\linWith}{ \basicLin^{\with} }
\newcommand{\linLolli}{ \basicLin^{\lollipop} }

\newcommand{\basicMon}{ \mathcal{O} }
\newcommand{\monTens}{ \basicMon^{\tens} }
\newcommand{\monWith}{ \basicMon^{\with} }
\newcommand{\monLolli}{ \basicMon^{\lollipop} }

\newcommand{\CL}{\mathsf{CL}}

\newcommand{\Syn}[1]{\mathrm{Syn}(#1)}

\newmvar{\type}{A \or B \or C \or D \or E} 
\newmvar{\groundtype}{\gamma \or \nu}

\newcommand{\length}[1]{| #1 |} 
\newmvar{\Ctx}{\Gamma \or \Delta \or \Theta \or \Sigma}

\newcommand{\app}[2]{#1 \, #2}
\newcommand{\appthree}[3]{#1 \, #2 \, #3}
\newcommand{\appfour}[4]{#1 \, #2 \, #3 \, #4}
\newcommand{\appfive}[5]{#1 \, #2 \, #3 \, #4 \, #5}

\newcommand{\letrule}[3]
		{\mathtt{let\:}{#1}\mathtt{\:be\:}{#2}\mathtt{\:in\:}{#3}}
\newcommand{\letruleforwith}[4]
		{\mathtt{let\:}{#1}\mathtt{\:be\:}\:\mathtt{p}_{#2} \mathtt{\:of\:}{#3}\mathtt{\:in\:}{#4}}
\newcommand{\abst}[2]{\uplambda{#1} \binddot {#2}}

\newenvironment{bprooftree} 
  {\leavevmode\hbox\bgroup}
  {\DisplayProof\egroup}

\newcommand{\unaryRule}[3]{\begin{bprooftree}
\AxiomC{\ensuremath{#1}}
\RightLabel{\scriptsize #3}
\UnaryInfC{\ensuremath{#2}}
\end{bprooftree}}

\newcommand{\binaryRule}[4]{\begin{bprooftree}
\AxiomC{\ensuremath{#1}}
\AxiomC{\ensuremath{#2}}
\RightLabel{\scriptsize #4}
\BinaryInfC{\ensuremath{#3}}
\end{bprooftree}}

\newcommand{\ternaryRule}[5]{\begin{bprooftree}
\AxiomC{\ensuremath{#1}}
\AxiomC{\ensuremath{#2}}
\AxiomC{\ensuremath{#3}}
\RightLabel{\scriptsize #5}
\TrinaryInfC{\ensuremath{#4}}
\end{bprooftree}}

%% file: pssemantics.tex
\newmvar{\interp}{s \or t \or l}
\newcommand{\sem}[1]{ \llbracket #1 \rrbracket }

%% file: psclonescategoriescl.tex
\newcommand{\w}{\mathrm{w}}

\newcommand{\sym}{\mathrm{sym}}

\newcommand{\altseq}[1]{ \seq{#1} }		
\newcommand{\altseqbig}[1]{ \seqbig{#1} }	

\newcommand{\CLw}{  \CL^{\mathrm{w}}  }
\newcommand{\CLwext}{  \CL^{\mathrm{wext}}  }

\DeclareMathOperator{\T}{\mathrm{T}}

\newcommand{\expCtx}[2]{ {#1 \To #2} }
\newcommand{\lolliCtx}[2]{ {[ #1 ; #2 ]} }
\newcommand{\lolliCtxbig}[2]{ {\big[ #1 ; #2 \big]} }
\newcommand{\itapp}[2]{ i_{#1 ; #2} }
\newcommand{\itappu}{ i }

\newcommand{\toCl}{ \mathrm{E} }
\newcommand{\toStForClone}{ \mathrm{S} }
\newcommand{\toStForCL}{ \mathrm{L} }
\newcommand{\toSKCat}{ \mathrm{N} }

\newcommand{\SKCatToClone}{ \mathrm{Cl} }
\newcommand{\SKCatToStClone}{ \mathrm{Q} }

\newcommand{\skcat}[1]
	{ (#1, \skforget, \lolliObj{-}{{=}}, \SS{}, \K{}{}, \catapp{})  }

\newcommand{\skforget}{U}

\newcommand{\skfunctor}[1]{ (#1, \homs, \forgs) }	
\newcommand{\homs}{\phi}
\newcommand{\forgs}{\psi}

\renewcommand{\SS}[1]{ {S}_{#1} }
\newcommand{\K}[2]{ K^{#1}_{#2} }
\newcommand{\catapp}[1]{ \epsilon_{#1} }

\newcommand{\sk}[1]
	{ (#1,  \lolliObj{-}{{=}}, \combS, \combK, \clapp{}{})  }

\newcommand{\comb}[1]{ \mathsf{#1} }
\newcommand{\combZ}{ \comb{Z} }
\newcommand{\combS}{ \comb{S} }
\newcommand{\combK}{ \comb{K} }

\newcommand{\combI}{ \comb{I} }
\newcommand{\combBprime}{ \comb{B'} }
\newcommand{\combB}{ \comb{B} }

\newcommand{\weq}{=_{\mathrm{w}}}
\newcommand{\wexteq}{=_{\mathrm{wext}}}

\DeclareMathOperator{\middot}{\cdot}
\newcommand{\clapp}[2]{ #1 \middot #2 }
\newcommand{\clappthree}[3]{  #1 \middot #2 \middot #3 }
\newcommand{\clappfour}[4]{  #1 \middot #2 \middot #3 \middot #4}

\newcommand{\close}[1]{ {#1}^{\mathrm{c}} }

%% file: title-and-abstract.tex
\title{Clones, closed categories, and combinatory logic\thanks{Supported by the 
Air Force Office of Scientific Research under award number {FA9550-21-1-0038}.} \\
	{\normalsize (longer version)}
}
%
%
\author{Philip Saville \orcidlink{0000-0002-8320-0280}}

\authorrunning{P. Saville}
%
\institute{Department of Computer Science, University of Oxford, UK 
\email{philip.saville@cs.ox.ac.uk}\\
\url{www.philipsaville.co.uk} 
}
\maketitle              
\begin{abstract}
We give an exposition of the semantics of the simply-typed $\uplambda$-calculus, and its linear and ordered variants, using multi-ary structures.
We define universal properties for multicategories, and use these to derive familiar rules for products, tensors, and exponentials. 
Finally we explain 
how to recover both the category-theoretic syntactic model and its semantic interpretation from the multi-ary framework.  \\[2mm]
We then use these ideas to study the semantic interpretation of combinatory logic and the simply-typed $\uplambda$-calculus without products. 
We introduce 
	\emph{extensional SK-clones}
and show these are sound and complete for both combinatory logic with 
extensional weak equality and the simply-typed $\uplambda$-calculus without 
products. We then show such SK-clones are equivalent to a variant of closed 
categories called \emph{SK-categories}, so the simply-typed 
$\uplambda$-calculus without products is the internal language of SK-categories.
As a corollary, we deduce that SK-categories have the same relationship to cartesian monoidal categories that closed categories have to monoidal categories.

\keywords{
		categorical semantics 
	\and abstract clones  
	\and lambda calculus 
	\and combinatory logic 
	\and closed categories 
	\and cartesian closed categories
}
\end{abstract}

%% file: sec-introduction.tex
Lambek's correspondence between cartesian closed categories and the 
simply-typed $\uplambda$-calculus is one of the central pillars of categorical 
semantics. 
One way of stating it categorically is to say that the syntax of typed $\uplambda$-terms over a \emph{signature} of base types and constants forms the free cartesian closed category 
	(for a readable overview, see~\cite{LambekAndScott,Crole1994}). 
The existence of this \emph{syntactic model} gives \emph{completeness}: if an equation holds in every model, it holds in the free one, and hence in the syntax. 
The free property then gives \emph{soundness}: for any interpretation of basic types and constants in a cartesian closed category ${\ccc\cat}$ one has a functor 
	$\sem{-}$
from the syntactic model to $\cat$, which is exactly the semantic interpretation of $\uplambda$-terms. 
The fact this functor is required to preserve cartesian closed structure amounts to showing that the semantic interpretation is sound with respect to the usual $\beta\eta$-laws.
All this justifies calling the simply-typed \mbox{$\uplambda$-calculus} the 
\emph{internal language} of cartesian closed categories.

This framework is powerful, but hides a fundamental mismatch: morphisms 
	$\type \to \type[2]$
in a category are \emph{unary}---they have just one input---but terms-in-context such as
	${\varx_1 : \type_1, \dots, \varx_n : \type_n \vdash \mmap : \type[2]}$
can have many inputs.
The standard solution~(\eg~\cite{Crole1994,Johnstone2002}) is to use 
categorical products to model contexts, so a term $\mmap$ as above corresponds 
to a map
	$\prod_{i=1}^n \type_i \to \type[2]$
out of the product.
	
Despite its evident success, this solution remains somewhat unsatisfactory, in two ways (see also~\cite{Jacobs1992}).
First, it forces us to conflate two different syntactic classes, namely contexts and product types.
As a result, some encoding is required to construct the syntactic model: the interpretation of
	$\varx : \type, \varx[2] : \type[2] \vdash \mmap : \type[3]$
is a term in context
	$p : \type \times \type[2]$.
This adds complexity to the construction, and results in the somewhat unintuitive fact that the semantic interpretation of a term $\mmap$ in the syntactic model may not be just $\mmap$ itself. 
%

Second, we are forced to include products in our type theory if we want a category-theoretic internal language---even though the calculus without products likely has a stronger claim to being called `the' simply-typed $\uplambda$-calculus
	(\eg~see Church's original definition~\cite{Church1940}). 
This raises the question: what categorical structure has the simply-typed $\uplambda$-calculus without products as its internal~language?

\paragraph*{\bf This paper.} This paper has three main aims.
	First, to explain how removing the mismatch between terms-in-context and morphisms outlined above clarifies the semantic interpretation of simply-typed $\uplambda$-calculi. To achieve this, one needs to move from the unary setting of categories to a \emph{multi-ary} setting, in which we have \emph{multimaps} $A_1, \dots, A_n \to B$. These ideas are not new, but are under-appreciated, and I hope this will provide self-contained introduction for a wider audience. 
	Second, to initiate a multi-ary investigation of the semantics of (cartesian) combinatory logic, in the style of Hyland's investigation of similar ideas for the untyped $\uplambda$-calculus~(\cite{Hyland2014,Hyland2015}).
	Finally, to use these results to define a categorical semantics for the simply-typed $\uplambda$-calculus without products.

\paragraph*{\bf Outline.} In \Cref{sec:clones-and-multicats,sec:universal-properties-for-multicategories,sec:cartesian-clones,sec:closed-clones,sec:cc-clones}, we explain how the multi-ary perspective yields a slick way to derive the unary semantic interpretation and syntactic model, together with soundness and completeness results (\Cref{sec:recovering-the-semantics-for-cartesian}). We also show how important type-theoretic constructions such as products and exponentials can be derived from the semantics. This framework accommodates different choices of structural rules, such as whether the language is ordered, linear, or cartesian. 

The idea of using multi-ary constructions goes back to Lambek~(\cite{Lambek1969,Lambek1989}), and
has recently been exploited to great effect in a very general setting by Shulman~\cite{Shulman2023}.
Particular cases can also be found in the works of Hyland~(\cite{Hyland2014,Hyland2015}), Hyland \& de Paiva~\cite{Hyland1993} and Blanco \& Zeilberger~\cite{Blanco2020}.
%
%
A reader familiar with these approaches will likely be unsurprised by the technical development below. However, we believe these ideas deserve to be more widely known, so spend time making them explicit in a concrete setting.

In \Cref{sec:SK-clones} we introduce a multi-ary model of (cartesian) combinatory logic, called \emph{SK-clones}, and prove that the sub-category of \emph{extensional} SK-clones is equivalent to the category of \emph{closed clones} modelling simply-typed $\uplambda$-calculus without products.  This provides a categorical statement of the classical correspondence between $\uplambda$-calculus and extensional combinatory logic (\eg~\cite{Barendregt1985,Gilezan1993}); it may be seen as a purely-semantic counterpart to the proof in~\cite{Altenkirch2023}, where Altenkirch~\etal~prove the generalised algebraic theories describing extensional combinatory logic and simply-typed $\uplambda$-calculus are the same.

Finally, in \Cref{sec:SK-categories} we introduce a version of Eilenberg \& Kelly's closed categories~(\cite{Eilenberg1966,Day1978}), called \emph{SK-categories}, and show that the category of SK-categories is equivalent to the category of extensional SK-clones, and so to the category of closed clones. Hence, SK-categories are a categorical model for the simply-typed $\uplambda$-calculus without products. SK-categories are a cartesian version of the prounital-closed categories of Uustalu, Veltri \& Zeilberger~(\cite{Uustalu2020,Uustalu2021}), which in turn are closely related to an (incomplete) suggestion of Shulman's~\cite{ShulmanNlab}.

Jacobs has also isolated a structure that is sound and complete for simply-typed $\uplambda$-calculus without products~\cite{Jacobs1992}. 
His approach, which fits into his elegant general framework~\cite{Jacobs1999}, is also predicated on a careful distinction between contexts and products. 
His models are certain indexed categories, with the contexts encoded by the 
indexing: this makes them feel closer to multi-ary structures. 
In SK-categories, by contrast, contexts are modelled within the category itself by using the closed structure (\cf~\cite[\S4.4]{Pitts2000}). 
Moreover, unlike other work relating closed categories to multi-ary structures, SK-categories do not force us to include a unit object in the corresponding type theory (\cf~\cite{Manzyuk2012}).

%

\paragraph*{\bf Technical preliminaries.}
For a set $\sorts$ we write 
	$\sorts^{\star}$
for the set of finite sequences over $\sorts$, and use Greek letters
	$\Ctx[1], \Ctx[2], \dots$ 
to denote elements of $\sorts^{\star}$.
The empty string is denoted $\diamond$, and the length of 
	$\Ctx$
by $\length\Ctx$.
Where the length of a sequence is clear, we write simply
	$\ind{\type}$
for
	$\type[1]_1, \dots, \type[1]_n$.
Contexts are assumed to be ordered lists.

We call multimaps of the form $\type \to \type[2]$ \emph{unary} and
a multimap 
	$\diamond \to \type[2]$ \emph{nullary}.

We define a \emph{signature} $\sig$ to be a set $\Obj\sig$ of 
	\emph{basic sorts} 
with sets
	$\sig(\Ctx; \type[2])$
of 
	\emph{constants}
		$c : \Ctx \to \type[2]$
for each 
	$\big( \Ctx , \type[2] \big) 
		\in \Obj{\sig}^{\star} \times \Obj\sig$.
A \emph{homomorphism of signatures}
	${\sigFunctor : \sig \to \sig'}$
is a map 
	$\Obj\sigFunctor : \Obj\sig \to \Obj{\sig'}$
with maps 
	$\sig(\type_1, \dots, \type_n; \type[2])
		\to
		\sig'(\sigFunctor\type_1, \dots, \sigFunctor\type_n; \sigFunctor\type[2])$
for each 
	$\big( (\type_1, \dots, \type_n), \type[2] \big) 
				\in \Obj{\sig}^{\star} \times \Obj\sig$.
We write $\Sig$ for the category of signatures and their homomorphisms.
One could also consider versions of higher-order constants, which may use the language's constructs. This extension does not change the theory significantly, and would require introducing multiple categories of signatures, so we do not seek this extra generality here.


We assume familiarity with the simply-typed $\uplambda$-calculus, as in~\eg~\cite{Crole1994}.
We denote the simply-typed $\uplambda$-calculus with constants and base types given by a signature $\sig$, and both product and exponential types modulo $\alpha\beta\eta$-equality, by $\stlpc_{\sig}$.
We write $\stpc_{\sig}$ and $\stlc_{\sig}$ for the fragments with just products and just exponentials, respectively.
Here we focus on the \emph{typed} cases: the untyped versions---both in the syntax and the multi-ary models---are recovered by fixing a single base type $\star$ such that $\Theta(\star, \dots, \star) = \star$ for each type constructor~$\Theta$.

We also assume familiarity with the basics of cartesian categories, cartesian closed categories, and monoidal categories, as in~\eg~\cite{cfwm,LambekAndScott}.
To avoid having to treat the unit type as a special case, cartesian categories 
are assumed to have $n$-ary products $\smallprod_n$ for all $n\in \Nat$. 
We also work with functors preserving structure \emph{strictly}: this 
simplifies the exposition without any great 
loss of generality. 
Thus, $\MonCat$, $\SymMonCat$ and $\CartCat$ denote the categories of monoidal categories, symmetric monoidal categories, and cartesian categories, respectively, with functors preserving all the data on the nose.

%% file: sec-multicategories-and-clones.tex
We begin with an intuitive overview of the place of multi-ary structures in semantics.
A multi-ary structure has \emph{multimaps} 
	$\type_1, \dots, \type_n \to \type[2]$
with multiple inputs and one output; unlike the morphisms in a category, multimaps correspond directly to terms-in-context.
As a result, it is often easier to construct a multi-ary free model than it is to construct a unary one, and the interpretation of a term-in-context $\mmap$ in the free model is given by $\mmap$ itself.
Moreover, every multi-ary structure gives rise to a unary one by restricting to multimaps with one input.
The multi-ary semantics therefore factors the unary one, as shown:
\vspace{.5mm}
\begin{equation}
	\label{eq:factoring-free-properties-schematically}
	\begin{tikzcd}[column sep = 1.2em, row sep = 1.5em]
		& \text{ multi-ary structures } \\
		\text{ signatures } && \text{ categorical structures }
		\arrow[
					""{name=0, anchor=center, inner sep=0}, 
					"\text{ free }"{}, 
					bend left = 8,
					 from=2-1, to=1-2, 
					start anchor={[shift={(13pt,0pt)}]north},
					end anchor={[shift={(4pt,2pt)}]south west}
		]
		\arrow[
				""{name=1, anchor=center, inner sep=0}, 
				"\text{ forget }"{}, 
				bend left = 8, 
				from=1-2, to=2-1,
				start anchor={[shift={(18pt,0pt)}]south west},
				end anchor={[shift={(-5pt,-3pt)}]north east}
		]
		\arrow[
						""{name=2, anchor=center, inner sep=0}, 
						"\text{ restrict to unary maps }", 
						bend left  = 8, 
						from=1-2, 
						to=2-3,
						start anchor={[shift={(0pt,2pt)}]south east},
						end anchor={[shift={(14pt,1pt)}]north west}						
					]
		\arrow[
						""{name=3, anchor=center, inner sep=0}, 
						"\text{ extend }", 
						bend left = 8,
						from=2-3, 
						to=1-2, 
						start anchor={[shift={(3pt,-2pt)}]north west},
						end anchor={[shift={(-12pt,1pt)}]south east}
				]
		\arrow["\vdash"{anchor=center, rotate=53}, draw=none, from=3, to=2]
		\arrow["\dashv"{anchor=center, rotate=-53}, draw=none, from=0, to=1]
	\end{tikzcd}
\end{equation}
\vspace{.5mm}

One can then `read off' the syntactic category, together with a guarantee that it has the right structure, by restricting the free multi-ary structure to unary maps.
Similarly, the usual semantic interpretation in (say) a cartesian closed category $\cat$ is exactly the interpretation that arises by extending $\cat$ to a multi-ary structure.
This gives an algebraic justification for encoding contexts as products: this 
is how one extends a cartesian closed category to a multi-ary structure.
(For the details of these points, see 
	\Cref{sec:recovering-the-semantics-for-cartesian}.)

The multi-ary perspective also provides a unifying framework for type theories with different structural rules. 
The simply-typed $\uplambda$-calculus is \emph{cartesian}: it admits the structural rules of weakening, contraction, and permutation 
	(as in~\eg~\cite[Fig. 3.2]{Crole1994}).  
The corresponding multi-ary structures are certain \emph{abstract clones}. 
\emph{Ordered} type theories (\eg~\cite{Lambek1958,Polakow1999}), also known 
as \emph{planar} type theories
(\eg~\cite{Abramsky2007,Zeilberger2015}), do not 
admit weakening, contraction, or 
permutation, and correspond to certain
\emph{multicategories}.
\emph{Linear} 
type theories (\eg~\cite{Girard1989}), 
which admit only permutation, 
correspond to certain
	\emph{symmetric multicategories}
(see also the alternative `tangled' option in~\cite{Mellies2018}).
Since abstract clones and symmetric multicategories may be seen as special 
cases of multicategories, we can develop a theory of how to add structure to 
cartesian, linear, and ordered type theories by analysing how to add structure 
to multicategories.

\subsection{Multicategories, clones, and their internal languages}

We now introduce multicategories and abstract clones and show how they correspond to certain type theories.  
An even more general framework for syntax, allowing multi-ary domains and codomains as well as both cartesian and linear contexts, is provided by Shulman's recent work with polycategories~\cite{Shulman2023}.
Clones, and their correspondence with syntax, also play a key role in the 
	`algebraic syntax'
programme of Fiore and collaborators initiated in~\cite{Fiore1999} (see \eg~\cite{Fiore2010,Arkor2020,Arkor2021}).
	
	
\begin{definition}[\cite{Lambek1969}]
A \emph{multicategory} $\multicat$ consists of a set $\Obj\multicat$ of \emph{objects} and sets 
	$\multicat(\Ctx; \type[2])$ 
of \emph{multimaps} for  every 
	$\Ctx \in \Obj{\multicat}^{\star}$ and $\type[2] \in \Obj{\multicat}$,
together with
	\begin{enumerate}
	\item 
		An \emph{identity} multimap
			$\Id_{\type} \in \multicat(\type; \type)$
		for every $\type \in \Obj{\multicat}$;
%
	\item 
		For any 
			$\type_1, \dots, \type_n, \type[2] \in \Obj{\multicat}$
		and 
			$(\Ctx[2]_i \in \Obj{\multicat}^{\star})_{i=1, \dots, n}$,
		a \emph{composition} map 
		\begin{align*}
			\multicat(\type_1, \dots, \type_n; \type[2])
				\times \smallprod_{\vari = 1}^n  \multicat(\Ctx[2]_i; \type_i)
				&\to  
				\multicat(\Ctx[2]_1, \dots, \Ctx[2]_n; \type[2]) \\
	 		\big( \mmap, (\mmap[2]_1, \dots, \mmap[2]_n) \big)
						&\mapsto 
						\msub{\mmap}{\mmap[2]_1, \dots, \mmap[2]_n}
		\end{align*}
%
	\end{enumerate}
subject to an associativity law and two unit laws 
	(see~\eg~\cite[p.~35]{Leinster2004}).
%
%
A \emph{multicategory functor}
	$\multicatFunctor : \multicat \to \multicat[2]$
consists of a map 
	$\Obj\multicatFunctor : \Obj\multicat \to \Obj{\multicat[2]}$ 
with maps
	$\multicatFunctor_{\ind{\type}, \type[2]} 
			:
			\multicat(\type_1, \dots, \type_n; \type[2])
			\to 
			\multicat[2](
				\multicatFunctor\type_1, \dots \multicatFunctor\type_n; 
				\multicatFunctor\type[2]
			)$	
	for every 
		$\type_1, \dots, \type_n, \type[2] \in \Obj{\multicat}$,
such that substitution and the identity are preserved
	(see~\eg~\cite[p.~39]{Leinster2004}).
%
%
\end{definition}

\begin{definition}[{\cite{May1972,Hyland1993}}]
	A \emph{symmetric multicategory} consists of a multicategory $\multicat$
	together with a symmetric group action: for each 
		$\type_1, \dots, \type_n \in \Obj{\multicat}$
	and
		$\sigma \in S_n$ 
	one has
		$(-) \act \sigma 
			: 
			\multicat(\type_1, \dots, \type_n; \type[2])
			\to 
			\multicat(\type_{\sigma 1}, \dots, \type_{\sigma n}; \type[2])
		$
	compatible with substitution and satisfying unit and associativity laws 
		(\eg~\cite[p. 54]{Leinster2004}).
	A \emph{symmetric multicategory functor} is a multicategory functor which preserves the action.
\end{definition}

We write 
	$\Multicat$
	(resp. $\SymMulticat$)
for the category of (symmetric) multicategories 
and their functors, and write 
	$\mmap : {\Ctx \to \type[2]}$
for 
	$\mmap \in \multicat(\Ctx; \type[2])$. 

\begin{example}
\label{ex:multicategories}
	\label{c:mon-to-multi-defined}
	Every monoidal category $(\cat, \tens, \tensu)$ induces a multicategory
			$\monToMulti{\cat}$. 
	The objects are those of $\cat$, with multimaps 
			$(\monToMulti{\cat})(\type_1, \dots, \type_n; \type[2])
				:= \cat(\bigotimes_{i=1}^n \type_i, \type[2])$
	for a chosen $n$-ary bracketing of the tensor product.
	This determines functors 
		$\MonCat \to \Multicat$,
	and
		$\SymMonCat \to \SymMulticat$ (see~\eg~\cite[p. 39]{Leinster2004});
	we denote both of these by $\monToMulti$.
\end{example}

Lambek~\cite{Lambek1969} essentially observed that every multicategory has an internal language, as follows. 
One identifies multimaps 
	$\mmap : \type_1, \dots, \type_n \to \type[2]$
with terms
	$\varx_1 : \type_1, \dots, \varx_n : \type_n \vdash {\mmap : \type[2]}$,
for a fixed ordering of an infinite set of variables $\{ \varx_1, \varx_2, \dots \}$.
The identity $\Id_{\type}$ is identified with the variable $\varx : \type$, and the composition operation becomes a formal substitution operation on the language.
Stated in this way, the three axioms become well-known properties of 
substitution: the unit laws say 
	$\csub{\varx}
				{\mmap[2]}
		= 
		\mmap[2]$
and 
	$\csub{\mmap}
				{\varx_1, \dots, \varx_n}
		= 
		\mmap$,
and the associativity law is a linear version of the so-called Substitution Lemma 
	(\eg~\cite[Lemma~2.1.16]{Barendregt1985}).


The next result shows this terminology does not differ too much from the notion of internal language in the introduction.
%
%
For a signature $\sig$ and 
	$\Ctx := ({\varx_i : \type_i})_{i=1, \dots, n}$,
write 
	$\basicMon_{\sig}$
for the ordered language generated by the two rules on the left below, and 
	$\basicLin_{\sig}$
for the linear language generated by all three rules:
	\begin{equation*}
	\label{eq:internal-language-of-a-multicategory}
	{\small
		\hspace{-1mm}
		%
		%
		\unaryRule
				{ \phantom{()\sig)} }
				{\varx : \type \vdash
						\varx : \type }
				{}
		\hspace{-1.5mm}
		\binaryRule 
			{c \in \sig(\Ctx; \type[2]) }
			{(\Ctx[2]_i \vdash \mmap[2]_i :  \type_i )_{i=1, \dots, n}}
			{\Ctx[2]_1, \dots, \Ctx[2]_n \vdash 
					c^{\S} 
						(\mmap[2]_1, \dots, \mmap[2]_n)
				: \type[2]
			}
			{}
		\hspace{-1.5mm}	
		\unaryRule
			{\Ctx[3], \varx : \type, \varx[2] : \type[2], \Ctx[2] \vdash \mmap 
			: \type[3]}
			{\Ctx[3], \varx[2] : \type[2], \varx : \type, \Ctx[2] \vdash \mmap 
			: \type[3]}
			{}
	}
	\end{equation*}
Substitution is defined as usual, so that the following rule is admissible:
\begin{equation} \label{eq:substitution}
	\binaryRule
		{\varx_1 : \type_1, \dots, \varx_n : \type_n \vdash \mmap : \type[2]}
		{( \Ctx[2]_i \vdash \mmap[2]_i : \type_i )_{i=1, \dots, n}}
		{ \Ctx[2]_1, \dots, \Ctx[2]_n \vdash 
			\csub
				{\mmap}
				{\mmap[2]_1 / \varx_1, \dots, \mmap[2]_n / \varx_n}
			: \type[2]
		}
		{}
\end{equation}

With this rule as composition, 
	$\basicMon_{\sig}$
and 
	$\basicLin_{\sig}$
define a syntactic multicategory 
	$\Syn{\basicMon_{\sig}}$
and a syntactic symmetric multicategory
	$\Syn{\basicLin_{\sig}}$,
respectively.
These define left adjoints to the functors
	$\Multicat \to \Sig$
and 
	$\SymMulticat \to \Sig$
 sending a (symmetric) multicategory $\multicat$ to the signature with objects
	$\Obj\multicat$
and constants
	$\big\{ \multicat(\Ctx; \type[2]) \big\}%
					_{\Ctx \in \Obj{\multicat}^{\star}, \type[2] \in \Obj\multicat}$;
we denote both these functors by $\forget$.

\begin{restatable}{lemma}{SemanticInterpOfBasicLin}
	$\Syn{\basicMon_{\sig}}$
		(resp. $\Syn{\basicLin_{\sig}}$)
	is the free multicategory (resp. symmetric multicategory) on $\sig$.
\end{restatable}

Thus, the internal language of a symmetric multicategory is the core of Abramsky's linear $\uplambda$-calculus~\cite{Abramsky1993}.
To recover a cartesian language, we use \emph{(multi-sorted) abstract clones}. 
These differ from multicategories in that the result of substituting
	$(\mmap[2]_i : \Ctx[2] \to \type_i)_{i=1, 2}$ 
into 
	$\mmap : \type_1, \type_2 \to \type[2]$
yields a multimap of type 
	$\Ctx[2] \to \type[2]$,
not
	$\Ctx[2], \Ctx[2] \to \type[2]$.
Abstract clones are equivalently
	\emph{cartesian multicategories}~(see \eg~\cite{Hyland2014}), 
but this formulation is less natural syntactically: it amounts to adding explicit duplication and deletion operations to the language.

\begin{definition}
	An \emph{abstract clone} $\clone$ consists of a set $\Obj\clone$ of \emph{sorts} and sets $\clone(\Ctx; \type[2])$ of \emph{multimaps} for every $\Ctx \in \Obj{\clone}^{\star}$ and $\type[2] \in \Obj{\clone}$, together with
\begin{enumerate}
	\item 
		\emph{Projection} multimaps 
			$\p{\vari}{\ind{\type}} \in \clone(\type_1, \dots, \type_n; \type_i)$
		for every 
					$\type_1, \dots, \type_n \in \Obj{\clone}$;
	\item 
		For every 
			$\type_1, \dots, \type_n, \type[2] \in \Obj{\clone}$
		and 
			$\Ctx[2] \in \Obj{\clone}^{\star}$,
		a \emph{substitution} operation 
		\begin{align*}
		\clone(\type_1, \dots, \type_n; \type[2])
				\times \smallprod_{\vari = 1}^n  \clone(\Ctx[2]; \type_i)
				&\to  
				\clone(\Ctx[2]; \type[2]) \\
		 \big( \mmap, (\mmap[2]_1, \dots, \mmap[2]_n) \big)
						&\mapsto 
						\csub{\mmap}{\mmap[2]_1, \dots, \mmap[2]_n}
		\end{align*}
%
	\end{enumerate}
subject to an associativity law and two unit laws for any 
	$\mmap \in \clone(\type_1, \dots, \type_n; \type[2])$,
	$\big( \mmap[2]_i \in \clone(\type[2]_1, \dots, \type[2]_m; \type_i) \big)_{i=1, \dots,n}$
and 
	$\big( \mmap[3]_j \in \clone(\Ctx[3]; \type[2]_j) \big)_{j =1 , \dots, m}$:
\begin{align*}
	\big(
				{\mmap}[\ind{\mmap[2]}]
			\big)[\ind{\mmap[3]}]
	=
	{\mmap}
		\big[
			\dots, 
			\csub
				{\mmap[2]_{\vari}}
				{\ind{\mmap[3]}},
			\dots 
		\big]		
\hspace{2mm}
,
\hspace{2mm}		
	\csub
		{\p{\vari}{\ind{\type}}}
		{\mmap[2]_1, \dots, \mmap[2]_n}
	=
	\mmap[2]_i
\hspace{2mm}
,
\hspace{2mm}  
	\csub
			{\mmap}
			{  \p{1}{\ind\type}, \dots, \p{n}{\ind\type} }
		=
		\mmap
\end{align*}
%
%
%
A \emph{homomorphism of clones}
	$\cloneFunctor : \clone \to \clone[2]$
consists of 
%
	a map $\Obj\cloneFunctor : \Obj\clone \to \Obj{\clone[2]}$ and
	 maps 
		$\cloneFunctor_{\ind{\type}, \type[2]} 
			:
			\clone(\type_1, \dots, \type_n; \type[2])
			\to 
			\clone[2](
				\cloneFunctor\type_1, \dots \cloneFunctor\type_n; 
				\cloneFunctor\type[2]
			)$
	for every
	$\type_1, \dots, \type_n, \type[2] \in \Obj{\clone}$,		
	such that 
%
	$
	\cloneFunctor(\p{i}{\ind{\type}})
		= \p{i}{\ind{(\cloneFunctor\type)}}$
and
%
%
	$\cloneFunctor{\big( 
		\csub
			{\mmap}
			{\mmap[2]_1, \dots, \mmap[2]_n} 
		\big)}
	=
		\csub
			{(\cloneFunctor\mmap)}
			{\cloneFunctor\mmap[2]_1, \dots, \cloneFunctor\mmap[2]_n} $.
%
We write 
	$\Clone$
for the category of clones and clone homomorphisms.
\end{definition}

\begin{example}[{\cf~\Cref{ex:multicategories}}]
	\label{ex:cartesian-cat-to-clone}
	%
	Any cartesian category 
		$\cartesian\cat$
	determines a clone 
		$\catToClone\cat$
	with sorts the objects of $\cat$ and 
		$(\catToClone{\cat})
			{\big( \type_1, \dots, \type_n ; \type[2] \big) }
			:= 
			\cat( \smallprod_{i=1}^n \type_i ; \type[2])$.
\end{example}	

We distinguish between clones and multicategories by using $[\dots]$ for a clone's substitution operation and $\seq{\dots}$ for a multicategory's composition operation. Every multicategory, and hence every clone, has an underlying category. 

\begin{definition}
	The \emph{nucleus}
		$\nucleus\multicat$
	of a multicategory or clone $\multicat$ is the category with the same objects 
	and 
		$\nucleus\multicat(\type, \type[2]) := \multicat(\type; \type[2])$.
	This defines functors 
		$\nucleus{(-)} : \Multicat \to \Cat$
	and 
		$\nucleus{(-)} : \Clone \to \Cat$
	to the category of small categories.
\end{definition}

%

The internal language of a clone is a cartesian version of that for 
multicategories. Write $\basicLambda_\sig$ for the language below; 
substitution is defined as usual.
\begin{equation*}
\label{eq:internal-language-of-a-clone}
	%
	%
	\unaryRule
			{ {\small (i = 1, \dots, n) } } 
			{\varx_1 : \type_1, \dots, \varx_n : \type_n  \vdash 
					\varx_i : \type_i }
			{}
	\hspace{0mm}
	\binaryRule 
		{c \in \sig(\Ctx; \type[2]) }
		{(\Ctx[2] \vdash \mmap[2]_i :  \type_i )_{i=1, \dots, n}}
		{\Ctx[2] \vdash 
				c^{\S} 
					(\mmap[2]_1, \dots, \mmap[2]_n)
			: \type[2]
		}
		{}
\end{equation*}
Identifying variables with projections, we get a syntactic clone $\Syn{\basicLambda_\sig}$.

\begin{lemma}
	\label{res:internal-language-of-a-clone}
	The canonical forgetful functor $\forget : \Clone \to \Sig$ has a left adjoint, and the free clone on $\sig$ is
		$\Syn{\basicLambda_\sig}$.
\end{lemma}

\begin{example}
	\label{ex:syntactic-models-of-cartesian-languages}
	The languages 
		$\stpc_{\sig}, \stlc_{\sig}$
	and 
		$\stlpc_{\sig}$
	each induce 
	{syntactic clones}
	we denote by 
		$\Syn{\stpc_{\sig}}, 
			\Syn{\stlc_{\sig}}$
	and 
		$\Syn{\stlpc_{\sig}}$,
	respectively.
\end{example}



%% file: sec-universal-properties-for-multicats.tex
 In this section we generalise the categorical notion of \emph{universal arrows} 
(as in \eg~\cite[\S3]{cfwm}) to give a notion of universal property for multicategories. This will provide a uniform way to introduce new connectives to a type theory. One could also define the required conditions directly (see~\cite{Blanco2020,Shulman2023}), but here we wish to emphasise that they arise from category-theoretic ideas. 
%

\begin{definition}[{\cf~\cite{Hermida2000}}]
	\label{def:universal-arrow}
	Let 
		$\multicatFunctor : \multicat \to \multicat[2]$
	be a multicategory functor. 
	\begin{enumerate}
	\item 
			A \emph{universal arrow from $\multicatFunctor$ to 
				$\objX[2] \in \Obj{\multicat[2]}$} 
			is a pair 
				$(\objR \in \Obj\multicat, \univ : \multicatFunctor\objR \to \objX[2])$
			such that for every 
				$\mmap : \multicatFunctor\type_1, \dots, \multicatFunctor\type_n \to \objX[2]$
			there exists a unique multimap 
				$\ext{\mmap} : \type_1, \dots, \type_n \to \objR$
			such that 
				$\msub{\univ}{\multicatFunctor(\ext\mmap)} = \mmap$.
		\item 
			A \emph{universal arrow from 
						$\objX_1, \dots, \objX_n \in \Obj{\multicat[2]}$ 
					to $\multicatFunctor$}
			is a pair 
				$(\objR \in \Obj{\multicat}, 
					\univ : \objX_1, \dots, \objX_n \to \multicatFunctor\objR)$
			such that for every 
				$\mmap : \objX_1, \dots, \objX_n \to \multicatFunctor\type[2]$
			there exists a unique multimap 
				$\ext{\mmap} : \objR \to \type[2]$
			such that 
				$\msub{\multicatFunctor(\ext{\mmap})}{\univ} = \mmap$.
	\end{enumerate}	
\end{definition}




We extend this definition---and hence our notion of universal property---to clones by using the next observation (\cf~the fact a cartesian category is~monoidal).

\begin{restatable}{lemma}{CloneToMulticat}
	\label{res:clones-to-multicats}
	There is a faithful functor 
		$\toMulti : \Clone \to \Multicat$
	sending a clone $\clone$ to the multicategory with the same objects and hom-sets, and composition given using substitution in $\clone$ and the projections.
\end{restatable}

\Cref{def:universal-arrow} does not involve `global' conditions like 
naturality, so is particularly amenable to a type-theoretic interpretation. 
As in the categorical setting, however, it can be rephrased using natural isomorphisms (\cf~\cite[\S3.2]{cfwm}).

\begin{lemma}
	\label{res:universal-arrows-as-natural-isomorphisms}
	Let 
		$\multicatFunctor : \multicat \to \multicat[2]$
	be a multicategory functor. 
	\begin{enumerate}
	\item 
		Giving a universal arrow 
		from $\multicatFunctor$ to 
			$\objX \in \Obj{\multicat[2]}$
		is equivalent to giving $\objR \in \multicat$ and an isomorphism 
			$\phi_{\ind\type} : 
				\multicat(\type_1, \dots, \type_n; \objR)
				\xra{\iso}
				\multicat[2](\multicatFunctor\type_1, \dots, \multicatFunctor\type_n; \objX[2])$,
		natural in the sense that the left diagram below commutes for any 
			$\mmap : \type_1, \dots, \type_n \to \type[2]$;
	\item 
		Giving a universal arrow 
		from 
			$\objX_1, \dots, \objX_n \in \Obj{\multicat[2]}$
		to $\multicatFunctor$ is equivalent to giving 
			$\objR \in \Obj{\multicat}$
		and an isomorphism 
			$\psi_{\type[2]} 
				: 
				\multicat(\objR; \type[2])
				\xra\iso
				\multicat[2](\objX_1, \dots, \objX_n; \multicatFunctor\type[2])$,
		natural in the sense that the right diagram below commutes for any 
			${\mmap[2] : \type[2] \to \type[3]}$.
	\end{enumerate}
	%
	%
	\[
		\begin{tikzcd}[column sep = 1em, scalenodes = .95]
			\multicat(\type[2]; \objR)
			\arrow{r}[]{\phi_{\type[2]}}
			\arrow{d}[swap]{  \msub{(-)}{\mmap}  }
			&
			\multicat[2](\multicatFunctor\type[2]; \objX)
			\arrow{d}{  \msub{(-)}{\multicatFunctor\mmap}  }
			\: 
			\\
			\multicat(\type_1, \dots, \type_n; \objR) 
			\arrow[swap]{r}{\phi_{\ind\type}}
			&
			\multicat[2](\multicatFunctor\type_1, \dots, \multicatFunctor\type_n; \objX)
		\end{tikzcd}
		\hspace{2mm}
		\begin{tikzcd}[column sep = 1em, scalenodes = .95]
			\multicat(\objR; \type[2])
			\arrow{r}[]{\psi_{\type[2]}}
			\arrow{d}[swap]{  \msub{\mmap[2]}{-}  }
			&
			\multicat[2](\objX_1, \dots, \objX_n; \multicatFunctor\type[2])
			\arrow{d}{  \msub{\multicatFunctor(\mmap[2])}{-}  }
			\: 
			\\
			\multicat(\objR; \type[3]) 
			\arrow[swap]{r}{\psi_{\type[3]}}
			&
			\multicat[2]\objX_1, \dots, \objX_n; \multicatFunctor\type[3])
		\end{tikzcd}		
	\]
\end{lemma}

A corollary is that
%
	giving a right adjoint 
	to a multicategory functor 
		$\multicatFunctor : \multicat[2] \to \multicat$
	in Hermida's 2-category of multicategories~\cite{Hermida2000}
	is equivalent to giving a mapping 
		$\multicatFunctor[2]_0 : \Obj{\multicat} \to \Obj{\multicat[2]}$
	and a universal arrow 
		$\multicatFunctor\multicatFunctor[2](\objX) \to \objX$
	from $\multicatFunctor$ to $\objX$ for each $\objX \in \Obj{\multicat[2]}$.

%% file: sec-cartesian-clones.tex
We now have enough to define products for multicategories, and hence for clones.
An $n$-ary product is exactly a limit over the discrete category with $n$ 
objects. 
Rephrasing in terms of universal arrows 
	(\eg~\cite[\S3]{cfwm})
we get that
equipping a category $\cat$ with $n$-ary products is exactly equipping it with a universal arrow from the diagonal functor
	$\Delta^{(n)} : \cat \to \cat^{\times n}$
to 
	$(\type_1, \dots, \type_n)$
for every 
	${\type_1, \dots, \type_n \in \cat}$.
	
Since $\Multicat$ has finite products defined in much the same way as the category of small categories $\Cat$,  we may make the following definition. 
The prefix `cartesian' is already used for multicategories, so we use 
	`{\bf f}inite-{\bf p}roducts'.

\begin{definition}
	\label{def:finite-product-multicategory}
	An \emph{fp-multicategory} is a multicategory $\multicat$ equipped with a universal arrow 
		$\big(
			\prod_{i=1}^n \type_i,
			(\pi_1^{\ind\type}, \dots, \pi_n^{\ind\type})
		\big)$
	from the diagonal functor
		$\Delta^{(n)}  : \multicat \to \multicat^{\times n}$
	to 
		$(\type_1, \dots, \type_n)$
	for every 
		$n \in \Nat$ 
	and
		$\type_1, \dots, \type_n \in \Obj{\multicat}$. 
\end{definition}

Asking for $\multicat$ to have finite products is equivalent to asking for a product object
	$\prod_{i=1}^n \type_i$
and unary multimaps 
	$\big( \pi_i^{\ind\type} : \prod_{i=1}^n \type_i \to \type_i \big)_{i=1, \dots, n}$
for each $\type_1, \dots, \type_n \in \Obj{\multicat}$, such that composition induces isomorphisms 
	$\multicat\big(\Ctx; \smallprod_{i=1}^n \type_i\big) 
		\iso 
		\prod_{i=1}^n \multicat(\Ctx; \type_i)$.
%
%
%
In the internal language, this amounts to the following rules:
\begin{equation}
	\label{eq:product-structure-spelled-out}
	\hspace{-3mm}
	\begin{aligned}
	\unaryRule	
		{  \phantom{ (\Ctx)}  }
		{p : \prod_{i=1}^n \type_i \vdash \pi^{\ind\type}_i(p) :  \type_i}
		{$(i = 1, \dots, n)$}
	\hspace{1mm}
	&,
	\hspace{1mm}
	\unaryRule
		{(\Ctx \vdash \mmap_i : \type_i)_{i=1, \dots, n}}
		{ \Ctx \vdash \altseq{\mmap, \dots, \mmap_n} : 
			\prod_{i=1}^n \type_i }
		{}	
	\\[3mm]
	\csubbig	
		{\pi_i^{\ind\type}(p)}
		{\altseq{ \mmap_1, \dots, \mmap_n}}
			= \mmap_i
	\hspace{3mm}
	&,
	\hspace{3mm}
	\altseqbig{ 
				\csub
						{\pi_1^{\ind\type}(p)}
						{\mmap[2]},
				\dots,  
				\csub
						{\pi_n^{\ind\type}(p)}
						{\mmap[2]}			
		}
		= 
		\mmap[2]
	\end{aligned}
\end{equation}
%

We can now derive the rules for $\with$ in linear $\uplambda$-calculus~\cite{Abramsky1993}. Indeed, given 
	$\Ctx, \varx : \type_i, \Ctx[3] \vdash \mmap : \type[2]$,
from~(\ref{eq:product-structure-spelled-out})
we get 
	$\Ctx, p : \prod_{i=1}^n \type_i, \Ctx[3] 
		\vdash \csub{\mmap}{\pi_i^{\ind\type}(p) / \varx} : \type[2]$.
This suggests the following.
Let $\monWith_{\sig}$ (resp. $\linWith_\sig$) be the extension of 
	$\basicMon_\sig$ (resp. $\basicLin_\sig$) with
\begin{equation*}
	\hspace{-4mm}
		\begin{aligned}
			\binaryRule 
				{\Ctx, \varx_i : \type_i, \Ctx[3] \vdash \mmap : \type[3]}
				{\Ctx[2] \vdash \mmap[2] : \with_{i=1}^n \type_i}
				{ \Ctx, \Ctx[2], \Ctx[3]
					\vdash \letruleforwith{\varx_i}{i}{u}{\mmap} 
					: \type[3]}
				{}		
			\hspace{2mm}
			,
			\hspace{2mm}
			\unaryRule 
				{(\Ctx \vdash \mmap_i : \type_i)_{i=1, \dots, n}}
				{ \Ctx 
					\vdash \seq{\mmap_1, \dots, \mmap_n} 
						: \with_{i=1}^n \type_i }
				{}
		\\[2mm]
		\letruleforwith
			{\varx_i}
			{i}
			{\seq{\mmap[2]_i}_{i=1}^n}
			{\mmap} 
		= 
			\csub{\mmap}{\mmap[2]_i / \varx_i}
		\hspace{2mm}
		,
		\hspace{2mm}
		\seq{
			\letruleforwith
				{\varx_i}
				{i}
				{\mmap[2]}
				{\varx_i}}_{i=1}^n
		= \mmap[2]
		\end{aligned}
\end{equation*}
where we write $\seq{\mmap[2]_i}_{i=1}^n$ for 
	$\seq{\mmap[2]_1, \dots, \mmap[2]_n}$.
%
%
This syntax defines a free property. To see this, say a multicategory functor $\multicatFunctor$ 
	\emph{(strictly) preserves finite products} 
if it preserves all the data on the nose, so that
	$\multicatFunctor(\smallprod_{i=1}^n \type_i)
			= \smallprod_{i=1}^n \multicatFunctor\type_i$,
	$\multicatFunctor( \pi_i^{\ind{\type}} )
			= \pi_i^{\ind{\multicatFunctor\type}}$,
and 
	$\multicatFunctor(\seq{\mmap_1, \dots, \mmap_n}) 
			= \seq{\multicatFunctor\mmap_1, \dots, \multicatFunctor\mmap_n}$.
%
Write
	$\fpMulticat$
for the category of fp-multicategories and product-preserving functors, and 
	$\fpSymMulticat$
for the subcategory of symmetric multicategories with finite products, with functors preserving both structures.

\begin{lemma}	
	\label{res:free-property-of-monWith-and-tensWith}
	The composite forgetful functor 
		$\fpMulticat \to \Multicat \to \Sig$
	has a left adjoint, and the free fp-multicategory on $\sig$ is $\Syn{\monWith_{\sig}}$.
	This extends to symmetric structure: replace $\fpMulticat$ by 
	$\fpSymMulticat$ and $\monWith$ by $\linWith$.	
\end{lemma}

Returning to the cartesian setting, we define products in a clone using the 
corresponding structure for multicategories and \Cref{res:clones-to-multicats}. 

\begin{definition}
	A \emph{cartesian clone} $\cartesian\clone$ is a clone $\clone$ equipped with a choice of finite products on $\toMulti\clone$.
	A \emph{(strict) homomorphism of cartesian clones} is a clone homomorphism 
		$\cloneFunctor$
	that strictly preserves all the product structure.
	We write $\CartClone$ for the category of cartesian clones and strict homomorphisms.
\end{definition}

Writing $\pi_i(\mmap)$ for the multimap
	$\csub	
		{\pi_i^{\ind\type}}
		{\mmap}$,
the rules~(\ref{eq:product-structure-spelled-out}) translate directly to the usual product rules of $\uplambda$-calculus. 
So cartesian clones exactly capture $\stpc$.

\begin{lemma}
	\label{res:free-property-of-stpc}
	The composite forgetful functor
		$\CartClone \to \Clone \to \Sig$
	has a left adjoint, and 
		$\Syn{\stpc_{\sig}}$	
	is the free cartesian clone on $\sig$.
\end{lemma}

Using the characterisation of universal arrows in terms of natural isomorphisms we get the following refinement of \Cref{ex:cartesian-cat-to-clone}.

\begin{example}
	\label{ex:cartesian-category-to-cartesian-clone}
	For any cartesian category $\cartesian\cat$ the induced clone
		$\catToClone\cat$
	is cartesian, essentially by definition; this extends to a functor
		$\catToClone : \CartCat \to \CartClone$.
	Moroever, if $\cartesian\clone$ is a cartesian clone, then so is its nucleus $\nucleus\clone$. Hence $\nucleus{(-)}$ restricts to a functor $\CartClone \to \CartCat$.
\end{example}

The two functors in this example are actually adjoints, yielding our first version of the schema in~(\ref{eq:factoring-free-properties-schematically}). The unit is identity-on-objects and sends 
	$\mmap : \type_1, \dots, \type_n \to \type[2]$
	to 
	$\csub{\mmap}{\pi_1^{\ind\type}, \dots, \pi_n^{\ind\type}} : \prod_{i=1}^n \type_i \to \type[2]$.

\begin{restatable}{proposition}{FactoringAdjunctionForCartesianClones}
	\label{eq:nucleus-left-adjoint-to-catToClone}
	\label{res:factoring-adjunction-for-cartesian-clones}
	The functor 
		$\nucleus{(-)} : \CartClone \to \CartCat$
	fits into the following diagram of adjunctions:
	\[
		\begin{tikzcd}[ampersand replacement=\&]
			\Sig \& \CartClone \& \CartCat
			\arrow[""{name=0, anchor=center, inner sep=0}, "\free"{}, yshift=2mm, from=1-1, to=1-2]
			\arrow[""{name=1, anchor=center, inner sep=0}, "\forget"{}, yshift=-2mm, from=1-2, to=1-1]
			\arrow[""{name=2, anchor=center, inner sep=0}, "\nucleus{(-)}", yshift=2mm, from=1-2, to=1-3]
			\arrow[""{name=3, anchor=center, inner sep=0}, "\catToClone", yshift=-2mm, from=1-3, to=1-2]
			\arrow["\adjUp"{anchor=center, rotate=0}, draw=none, from=3, to=2]
			\arrow["\adjUp"{anchor=center, rotate=-0}, draw=none, from=0, to=1]
		\end{tikzcd}
	\]
	Moreover,
			$\forget \circ \catToClone$
	is equal to the canonical forgetful functor $\CartCat \to \Sig$.
	Hence, the free cartesian category on $\sig$ is canonically isomorphic to 
		$\nucleus{\Syn{\stpc_{\sig}}}$.
\end{restatable}

%% file: sec-cartesian-from-representability.tex
In the preceding section we defined products using a multi-ary version of the familiar universal property. 
There is another way to define `monoidal structure' in a multicategory: Hermida's
	\emph{representability}~\cite{Hermida2000}.
From the perspective of linear logic, the finite product structure explored above corresponds to the additive conjunction $\with$; Hermida's representability will correspond to the multiplicative conjunction $\tens$.
We shall also see that, for clones, the two are~equivalent.
%


\begin{definition} 
	\label{def:representable-multicategory}
	A \emph{representable multicategory} is a multicategory $\multicat$ equipped with
	a universal arrow 
		$\Big( 
					{\T}(\objX_1, \dots, \objX_n), 
					\univ_{\ind{\objX}} : \objX_1, \dots, \objX_n \to \T(\objX_1, \dots, \objX_n)
			\Big)$ 
	from 
		$\objX_1, \dots, \objX_n$
	to the identity $\id_{\multicat}$ for each 
		$\objX_1, \dots, \objX_n \in {\Obj\multicat}$; 
	we write 
		$\T_{i=1}^n \objX_i$
	for 
		$\T(\objX_1, \dots, \objX_n)$.
	These universal arrows must be closed under composition, so
		\[
			\objX_1, \dots, \objX_n, \objX[2]_1, \dots, \objX[2]_m
				\xra{ \seq{ \univ_{\ind{\objX}}, \univ_{\ind{\objX[2]}}  }}
			\T_{i=1}^n \objX_i, \T_{j=1}^m \objX[2]_j
				\xra{\univ}
			{\T}{\big(  {\T}_{i=1}^n \objX_i, \T_{j=1}^m \objX[2]_j \big)}
		\]
		must also be universal.
		A \emph{representable multicategory functor} $\multicatFunctor$ is a multicategory functor that preserves all the universal arrows, so that 
			$\multicatFunctor(\T_{i=1}^n \type_i)
				= \T_{i=1}^n \multicatFunctor\type_i$,
			$\multicatFunctor(\univ_{\ind{\type}})
				= \univ_{\multicatFunctor\ind{\type}}$
		and 
			$\multicatFunctor(\ext\mmap) = \ext{\multicatFunctor\mmap}$.
		Write $\RepMulticat$ for the category of representable multicategories, and $\SymRepMulticat$ for the category of representable multicategories whose underlying multicategories are also symmetric, with functors preserving both structures.
\end{definition}

\begin{example}[{\cf~\Cref{ex:multicategories}}]
	\label{ex:monoidal-cat-to-representable-multicategory}
	The multicategory $\monToMulti\cat$ induced by a monoidal category 
		$\monoidal{\cat}$
	is representable. 
	We therefore obtain functors 
		$\MonCat \to \RepMulticat$ 
	and
		$\SymMonCat \to \SymRepMulticat$;
	we denote them both $\monToMulti$.
\end{example}

A representable multicategory is a multicategory equipped with rules which are dual to those in~(\ref{eq:product-structure-spelled-out}) in the sense that the universal arrow goes the other direction.
%
%
%
Indeed, writing 
	$\varx_1 \tens \dots \tens \varx_n$
for $\univ_{\ind\type}$, and 
	$\letrule{(\varx_1, \dots, \varx_n)}{p}{\mmap}$
for $\ext{\mmap}$, and extending this to all terms by 
\begin{equation*}
	\begin{aligned}
		 \mmap[2]_1 \tens \dots \tens \mmap[2]_n 
		 &:=
			\csub
			 	{(\varx_1 \tens \dots \tens \varx_n)}
			 	{\mmap[2]_1 / \varx_1, \dots, \mmap[2]_n / \varx_n}
		\\[2mm]
		\letrule{(\varx_1, \dots, \varx_n)}{\mmap[2]}{\mmap}
		&:=
			\csub 
				{\big(  \letrule{(\varx_1, \dots, \varx_n)}{p}{\mmap}  \big)}
				{  \mmap[2] / p  }
	\end{aligned}
\end{equation*}
we obtain the following rules, where 
	$\Ctx := (\varx_i : \type_i)_{i=1, \dots, n}$:
\begin{equation} \label{eq:rules-for-otimes}
	\hspace{-2mm}
	\small
	\unaryRule
			{ ( \Ctx[2]_i \vdash \mmap[2]_i : \type_i)_{i=1, \dots, n} }
			{ \Ctx[2]_1, \dots, \Ctx[2]_n 
					\vdash\otimes_{i=1}^n \mmap[2]_i
					:  \bigotimes_{i=1}^n \type_i } 
			{}
	\hspace{0mm}
	,
	\hspace{0mm}
	\binaryRule
			{\Lambda, \Ctx, \Ctx[3] \vdash \mmap : \type[2]  }
			{	\Ctx[2] \vdash \mmap[2] : \bigotimes_{i=1}^n \type_i }
			{ \Lambda, \Ctx[2], \Ctx[3] \vdash 
					\letrule{(\varx_1, \dots, \varx_n)}{\mmap[2]}{\mmap} : \type[2] }
			{}
\end{equation}
\begin{equation*}	
	\hspace{-2mm}
	\letrule
		{(\varx_1, \dots, \varx_n)}
		{p}
		{\csub
			{\mmap}
			{\otimes_{i=1}^n \varx_i / p}
		}
	= 
		\mmap 
	\hspace{1mm}
	,
	\hspace{1mm}
	\letrule
		{(\varx_1, \dots, \varx_n)}
		{\otimes_{i=1}^n \varx_i}
		{\mmap}
	= 
	\mmap
\end{equation*}

We write 
	$\monTens_{\sig}$
	(resp. $\linTens_\sig$)
for the extension of 
	$\basicMon_{\sig}$
	(resp. $\basicLin_{\sig}$)
with these rules. 
This is essentially the tensor fragment of Abramsky's linear $\uplambda$-calculus~\cite{Abramsky1993}. 
The connection with multicategories was already made in by 
	Hyland~\&~de Paiva~\cite{Hyland1993}, 
who showed this type theory arises from Lambek's 
	\emph{monoidal multicategories}~\cite{Lambek1989}.

\begin{lemma}
	The composite forgetful functor
		$\RepMulticat \to \Multicat \to \Sig$
	has a left adjoint, and the free representable multicategory on $\sig$ is the syntactic multicategory $\Syn{\monTens_\sig}$.
	The same holds for symmetric structure, if one replaces $\RepMulticat$	by $\SymRepMulticat$ and $\monTens$ by $\linTens$.
\end{lemma}

Combining this lemma with \Cref{res:free-property-of-monWith-and-tensWith}, one sees that a multicategory equipped with representable and finite-product structure corresponds to a linear type theory with both $\tens$ and $\with$. 
%

We can also obtain a linear version of \Cref{res:factoring-adjunction-for-cartesian-clones}.
Hermida~\cite{Hermida2000} showed that the 2-category of representable multicategories is 2-equivalent to the \mbox{2-category} of monoidal categories, and Weber showed this extends to the symmetric case~\cite{Weber2013}. 
From these constructions one can extract functors 
	$\monToMulti : \RepMulticat \to \MonCat$
and 
	$\monToMulti_\sym : \SymRepMulticat \to \SymMonCat$
sending a (symmetric) representable multicategory to a (symmetric) monoidal structure on its nucleus, together with equivalences
	$\RepMulticat \simeq \MonCat$
and
	$\SymRepMulticat \simeq \SymMonCat$.
So we get the following.

\begin{proposition}
	\label{res:factoring-monoidal-categories-interpretation}
	The functors 
		$\multiToMon$
	and 
		$\multiToMon_{\sym}$ 
	fit into the following diagram of adjunctions, where in each case the right-hand adjunction is an equivalence:
	\[
		\begin{tikzcd}[column sep = 0em, row sep = 1.5em]
			& \RepMulticat
				\arrow[phantom,
							start anchor={[shift={(-1pt,2pt)}]south east},
							end anchor={[shift={(-12pt,1pt)}]north}
							]{dr}[marking, yshift=-1.5mm, xshift=-1mm]
						{{\scriptstyle {\smalladjUp}\simeq}} \\
			\Sig && \MonCat
			\arrow[""{name=0, anchor=center, inner sep=0}, "\free"{}, bend left 
			= 6, 
			xshift = -3mm, from=2-1, to=1-2]
			\arrow[""{name=1, anchor=center, inner sep=0}, "\forget"{}, bend left = 12, from=1-2, to=2-1]
			\arrow[""{name=2, anchor=center, inner sep=0}, "\multiToMon", bend 
			left = 20, from=1-2, to=2-3, 
				start anchor={[shift={(-1pt,2pt)}]south east},
				end anchor={[shift={(-12pt,1pt)}]north}					
			]
			\arrow[""{name=3, anchor=center, inner sep=0}, "\monToMulti", bend 
			left = 20, from=2-3, to=1-2, shift left = 1mm, 
				start anchor={[shift={(-22pt,1pt)}]north},
				end anchor={[shift={(-10pt,2pt)}]south east}
			]
			\arrow["\dashv"{anchor=center, rotate=-45}, draw=none, from=0, to=1]
		\end{tikzcd}
		\hspace{.4cm}
		\begin{tikzcd}[column sep = -.8em, row sep = 1.5em ]
			& \SymRepMulticat 
			\arrow[phantom,
						start anchor={[shift={(-1pt,2pt)}]south east},
						end anchor={[shift={(-12pt,1pt)}]north}
						]{dr}[marking, yshift=-2mm, xshift=-1mm]
					{{\scriptstyle {\smalladjUp}\simeq}} \\
			\Sig && \SymMonCat
			\arrow[""{name=0, anchor=center, inner sep=0}, 
			"\free"{yshift=-1.5mm, 
			xshift=-1mm}, bend left 
			= 17, from=2-1, to=1-2]
			\arrow[""{name=1, anchor=center, inner sep=0}, "\forget"{}, bend 
			left = 14, from=1-2, to=2-1,
				start anchor={[shift={(-18pt,-1pt)}]south}
			]
			\arrow[""{name=2, anchor=center, inner sep=0}, 
			"\multiToMon_{\sym}", bend left  = 12, from=1-2, to=2-3,
				start anchor={[shift={(-1pt,2pt)}]south east},
				end anchor={[shift={(-14pt,1pt)}]north}				
				]
			\arrow[""{name=3, anchor=center, inner sep=0}, 
			"\monToMulti_{\sym}", bend left = 12, from=2-3, to=1-2,
				xshift=-2mm,
				start anchor={[shift={(-22pt,1pt)}]north},
				end anchor={[shift={(-10pt,2pt)}]south east}			
			]
			\arrow["\dashv"{anchor=center, rotate=-40}, draw=none, from=0, to=1]
		\end{tikzcd}
	\]
	Moreover,
			$\forget \circ \monToMulti$
	and 
		$\forget \circ \monToMulti_\sym$
	are both equal to the canonical forgetful functor to $\Sig$.
	Hence, the free monoidal 
		(resp. symmetric monoidal)	
	category on a signature $\sig$ is canonically isomorphic to 
		$
		\multiToMon{\big(\Syn{\monTens_\sig}\big)}
		$
		(resp. $\multiToMon{\big(\Syn{\linTens_\sig}\big)}$). 
\end{proposition}

%
We now turn to studying representability in the cartesian setting.

\begin{definition}
	A \emph{representable clone} is a clone $\clone$ equipped with a choice of representable structure on $\toMulti\clone$.
	A \emph{representable clone homomorphism} is a clone homomorphism  which preserves the representable structure 
	as in \Cref{def:representable-multicategory}.
\end{definition}

A cartesian clone makes the \emph{projections} primitive
	(recall~(\ref{eq:product-structure-spelled-out})), 
but a representable clone makes the \emph{pairing operation} primitive
	(recall~(\ref{eq:rules-for-otimes})).
It turns out these perspectives are equivalent. In the proof-theoretic setting such ideas are well-studied (\cf~the equivalence of G-systems and N-systems in~\cite[\S 3.3]{Troelstra2000}); the categorical statement has also been made by Pisani~\cite{Pisani2014} and Shulman~\cite{Shulman2023}.
%
	
\begin{restatable}{proposition}{RepresentableEqualsCartesian}
	\label{res:representable-equals-cartesian}
	Equipping a clone $\clone$ with representable structure is equivalent to equipping $\clone$ with cartesian structure. 
\end{restatable}

In 
	\Cref{res:factoring-monoidal-categories-interpretation}
we gave an equivalence of categories 
but in 
	\Cref{res:factoring-adjunction-for-cartesian-clones}
we only gave an adjunction. 
We can now upgrade~the latter to an equivalence. Indeed, $\nucleus{(-)} \circ \catToClone$ is equal to the identity. On the other hand, if $\cartesian\clone$ is a cartesian clone then by \Cref{res:representable-equals-cartesian} and \Cref{res:universal-arrows-as-natural-isomorphisms} we have a multi-natural isomorphism 
	$\clone(\type_1, \dots, \type_n; \type[2]) \iso \clone(\prod_{i=1}^n \type_i; \type[2]) = \catToClone(\nucleus\clone)(\type_1, \dots, \type_n; \type[2])$.


\begin{restatable}{corollary}{CartCloneEquivalentToCartCat}
	\label{res:cartclone-equivalent-to-cartcat}
	The functors $\catToClone$ and $\nucleus{(-)}$ of
		\Cref{res:factoring-adjunction-for-cartesian-clones}
	define an adjoint equivalence 
		$\CartClone \simeq \CartCat$.
\end{restatable}

%% file: sec-soundness-and-completeness.tex
We now show how the usual semantic interpretation, syntactic model, and soundness and completeness results can be derived from the multi-ary framework. 
Although we shall not pursue the point in detail for reasons of space, essentially the same argument holds for all the calculi considered in this paper.

\paragraph*{\bf Semantic interpretation and soundness.}
We recover the usual  semantic interpretation of $\stpc$ in a cartesian category by
	\Cref{res:free-property-of-stpc}
and
	\Cref{ex:cartesian-category-to-cartesian-clone} as follows. 
Let $\forget : \CartCat \to \Sig$ be the functor sending a cartesian category 
	$\cartesian\cat$
to the signature with objects those of $\cat$ and constants
	$\big\{ \cat(\prod_{i=1}^n \type_i, \type[2]) \big\}_{\type_1, \dots, \type_n, \type[2] \in \cat}$.
An interpretation 
	${\interp : \sig \to \forget\cat}$
of basic types and constants in $\cat$ is exactly an interpretation 
	$\interp : \sig \to \forget(\catToClone\cat)$
in the induced cartesian clone. 
The unique extension $\interp\sem{-} : \Syn{\stpc_\sig} \to \catToClone\cat$ sends a term
	$\varx_1 : \type_1, \dots, \varx_n : \type_n \vdash \mmap : \type[2]$
to a multimap 
	$\interp\sem{\varx_1 : \type_1, \dots, \varx_n : \type_n} \to \interp\sem{\type[2]}$
in $\catToClone\cat$, which is exactly a map
	$\prod_{i=1}^n \interp\sem{\type_i} \to \interp\sem{\type[2]}$
in $\cat$. It is not hard to show this coincides with the usual, inductively defined semantic interpretation.
Unlike with the unary approach, we do not need to prove soundness with respect to $\beta\eta$ as a separate lemma: this holds immediately from the fact $\interp\sem{-}$ is a cartesian clone homomorphism. 

Moreover, for any objects $\type_1, \dots, \type_n$ in a cartesian clone one can construct a `multi-isomorphism' 
	$(\type_1, \dots, \type_n) \iso \prod_{i=1}^n \type_i$ 
	(see~\cite[Lemma 4.2.16]{mythesis}).
Hence, in a cartesian simple type theory with products, \emph{contexts must coincide with product types}. Together with the preceding, this provides a mathematical explanation for the identification of contexts and product types in the interpretation of $\stlpc$.

For a further example of a soundness property which normally requires a separate, inductively proven, lemma consider the following (\cf~\cite[Lemma~3.5.2]{Crole1994}). 
Fix a context $\Ctx := (\varx_i : \type_i)_{i=1, \dots, n}$ and a term
	$\Ctx \vdash \mmap : \type[2]$.
Then $\mmap$ weakens to
	$\Ctx, y : \type_{n+1} \vdash \mmap^y : \type[2]$, 
which is exactly the term obtained from the substitution rule~(\ref{eq:substitution}) using the $n$ variables $(\Ctx, y : \type_{n+1}  \vdash \varx_i : \type_i)_{i=1, \dots, n}$, so 
	$\mmap^y = \mmap{[\varx_1 / \varx_1, \dots, \varx_n / \varx_n]}$.
In other words, in the clone $\Syn{\stpc_{\sig}}$ we have that 
	$\mmap^y = \csub{\mmap}{\p{1}{\Ctx, \type_{n+1}}, \dots, \p{n}{\Ctx, \type_{n+1}}}$.
By the clone structure of $\catToClone\cat$ and the fact any clone homomorphism preserves substitution, it follows that
	$\interp\sem{\mmap^y} : \prod_{i=1}^{n+1} \interp\sem{\type_i} \to \interp\sem{\type[2]}$ is
$
	\interp\sem{\mmap^y}
		= \interp\sem{\csub{\mmap}{\p{1}{\Ctx, \type_{n+1}}, \dots, \p{n}{\Ctx, \type_{n+1}}}}
		= \interp\sem{\mmap} \circ \seq{\pi_1, \dots, \pi_n}
$. 

\paragraph*{\bf Syntactic model.} 
We extract the construction from \Cref{res:factoring-adjunction-for-cartesian-clones}.
For a signature $\sig$ the cartesian category
	$\nucleus{\Syn{\stpc_{\sig}}}$
has objects the types of $\stpc_\sig$ and morphisms
	$\type \to \type[2]$
given by $\alpha\beta\eta$-equivalence classes of terms 
	$\varx : \type  \vdash \mmap : \type[2]$
for a fixed variable $\varx$. 
%
%
Composition is substitution and the identity on $\type$ is the variable $\varx : \type$. 
%
%
The projections are 
	$\varx : \prod_{i=1}^n \type_i \vdash \pi_i^{\ind\type}(x) : \type_i$
and the pairing of the maps
	$( 
			\varx : \type[3] \vdash \mmap_i : \type_i 
		)_{i=1, 2}$
is
	$\varx : \type[3] \vdash \seq{\mmap_1, \mmap_2} : \type_1 \times \type_2$.
%
The usual proofs that this is indeed cartesian (see~\eg~\cite[Chapter 3]{Crole1994}) have been replaced by the simple observation of \Cref{ex:cartesian-category-to-cartesian-clone}.

\paragraph*{\bf Completeness.}
Once again, the proof is largely category-theoretic.
Note first that the functor $\nucleus{(-)} : \CartClone \to \CartCat$ is faithful. One can prove this directly using \Cref{res:representable-equals-cartesian} or infer it from  \Cref{res:cartclone-equivalent-to-cartcat} and the fact any equivalence is fully faithful. In any case, it follows by standard results (\eg~\cite[Lemma 4.5.13]{Riehl2016}) that the unit $\eta'$ of the adjunction $\nucleus{(-)} \dashv \catToClone$ is monic. 
Just as in $\Cat$, any monomorphism of clones is injective on objects and injective on multimaps.
It suffices, therefore, to find a semantic interpretation 
	$\iota\sem{-}$
which is equal to a component of
	$\eta'$.
This is accomplished by the next lemma.

\begin{lemma}
	Let 
		\begin{tikzcd}
			\cat 
						\arrow[yshift=1.2mm]{r}{\functor}
						\arrow[phantom]{r}{\adjUp} &
			\cat[2] 
						\arrow[yshift=-1.2mm]{l}{\forget}
						\arrow[yshift=1.2mm]{r}{\functor'}
						\arrow[phantom]{r}{\adjUp} &
			\cat[3]	\arrow[yshift=-1.2mm]{l}{\forget'}
		\end{tikzcd}
	be adjunctions with units $\eta : \id_{\cat} \To \forget\functor$ and $\eta' : \id_{\cat[2]} \To \forget'\functor'$. 
	Then for any $C \in \cat$, the unit $\eta'_{\functor C} : \functor C \to \forget'\functor'\functor C$ is the unique map $h$ such that the following diagram commutes:
	\begin{td}
		\forget\functor C
			\arrow[dashed]{r}{\forget h} &
		\forget\forget'\functor'\functor C
		\\
		C 
			\arrow{u}{\eta_C}
			\arrow[swap]{r}{\eta_C} &
		\forget\functor C
			\arrow[swap]{u}{\forget\eta'_{\functor C}}
	\end{td}
\end{lemma}

In the setting of \Cref{eq:nucleus-left-adjoint-to-catToClone} this lemma implies that the component 
		$\eta'_{\functor \sig} : \Syn{\stpc_\sig} \to \catToClone\big( \nucleus{\Syn{\stpc_\sig}} )$
of the unit for the adjunction $\nucleus{(-)} \dashv \catToClone$ is exactly the unique cartesian clone homomorphism $\iota\sem{-}$ extending the obvious interpretation 
	$\iota := \sig \hookrightarrow  \nucleus{\Syn{\stpc_\sig}}$
of base types and constants in the free cartesian category. By our preceding discussion, this clone homomorphism is injective on multimaps: so if $\iota\sem{t} = \iota\sem{t'}$ then $t = t'$ in $\Syn{\stpc_\sig}$, hence $t =_{\beta\eta} t'$.

%% file: sec-closed-clones.tex
To define closed structure, we follow Lambek's definition and simply upgrade the hom-set definition of exponentials to multicategories.

\begin{definition}[{\cite{Lambek1989}}]
	\label{def:closed-multicategory}
	A \emph{closed multicategory} is a multicategory $\multicat$ equipped with an object 
		$\lolliObj{\type}{\type[2]}$
	and multimap 
		$\eval_{\type, \type[2]} :  \lolliObj{\type}{\type[2]}, \type \to \type[2]$
	for every $\type, \type[2] \in \Obj\multicat$, such that composition induces isomorphisms as shown:
	\begin{equation} \label{eq:closed-structure}
		\begin{tikzcd}[column sep = 4em]
			\multicat(\Ctx, \type; \type[2])			 
			 & 			
			\multicat(\Ctx; \lolliObj{\type}{\type[2]})
			\arrow[""{name = 0},
					"\Lambda^{\type}",
					yshift=2mm, from=1-1, to=1-2]
			\arrow[""{name=1, anchor=center, inner sep=0}, yshift=-2mm, 
							"\msub
										{\eval_{\type, \type[2]}}
										{(-), \Id_{\type}}", from=1-2, to=1-1]
			\arrow["\iso"{description}, draw=none, from=0, to=1]
		\end{tikzcd}
	\end{equation}
	A \emph{(strict) closed multicategory functor} is a multicategory functor $\multicatFunctor$ which preserves all the data:
		$\multicatFunctor(\lolliObj{\type}{\type[2]})
					= \lolliObj{\multicatFunctor\type}{\multicatFunctor\type[2]}$,
		$\multicatFunctor(\eval_{\type, \type[2]})
					= \eval_{\multicatFunctor\type, \multicatFunctor\type[2]}$
	and 
		$\multicatFunctor(\Lambda \mmap)
					= \Lambda(\multicatFunctor\mmap)$.
	We write 
		$\ClMulticat$
	for the category of closed multicategories and their functors,
	and 
		$\ClSymMulticat$
	for the category of symmetric multicategories with closed structure, and functors preserving both of these. 
\end{definition}

\begin{example}
		If
			$\monoidalclosed{\cat}$
		is a closed (symmetric) monoidal category  then the induced (symmetric) multicategory
			$\monToMulti\cat$
		is also closed.
\end{example}

Closed multicategories allow us to model exponentials without requiring a tensor product.
%
%
%
Writing out the rules in the internal language, we get the map $\Lambda^{\type}$ in (\ref{eq:closed-structure}) as the usual abstraction rule, and the evaluation map as the application
		$f : \type \lollipop \type[2], x : \type 
					\vdash \app{f}{x} : \type[2]$.
We then see that
	$\Ctx[2], f : \type \lollipop \type[2], \varx : \type 
			\vdash u[\app{f}{x} / y] : \type[3]$
whenever
	$\Ctx[2], y : \type[2] \vdash u : \type[3]$,
so we recover a small adaptation of Abramsky's rules for exponentials.
Write
	$\monLolli_\sig$
	(resp. $\linLolli_\sig$)
for the extension of $\basicMon_{\sig}$ (resp. $\basicLin_\sig$) with the following rules and the $\beta\eta$-laws familiar from $\stlc$:
\begin{equation*}
	\begin{aligned}
			\ternaryRule
					{\Ctx[2], \varx[2] : \type[2] \vdash \mmap[2] : \type[3]}
					{\Ctx[3] \vdash \mmap : \type \lollipop \type[2]}
					{\Ctx \vdash \mmap[3] : \type}					
					{\Ctx[2], \Ctx[3], \Ctx
							\vdash 
							\csub
								{\mmap[2]}
								{ \app{\mmap}{\mmap[3]} / \varx[2] }
							: \type[3]
					}
					{}
			\hspace{0mm}
			&,
			\hspace{0mm}
			\unaryRule
		 		{ \Ctx, \varx : \type \vdash \mmap : \type[2] }
				{ \Ctx \vdash \abst{\varx}{\mmap} : \type \lollipop \type[2]  }
				{}
	\end{aligned}
\end{equation*}

\begin{lemma}[{\cite{Hyland1993}}]
	The composite forgetful functor
		$\ClMulticat \to \Multicat \to \Sig$
	has a left adjoint, and the free closed multicategory on $\sig$ is the syntactic multicategory $\Syn{\monLolli_\sig}$.
	The same holds for symmetric structure, if one replaces $\ClMulticat$	by $\ClSymMulticat$ and $\monLolli$ by $\linLolli$.
\end{lemma}

For the cartesian case, we follow the same procedure as in \Cref{sec:cartesian-clones}. 

\begin{definition}
	\label{def:closed-clone}
	A \emph{closed clone} is a clone $\clone$ equipped with a closed structure on $\toMulti\clone$.
	We write $\ClClone$ for the category of closed clones and clone homomorphisms preserving the closed structure as in \Cref{def:closed-multicategory}.
\end{definition}

\begin{example}
	\label{ex:ccc-to-closed-clone}
	If $\ccc\cat$ is a cartesian closed category, the clone 
		$\catToClone\cat$
	is closed.
\end{example}

\Cref{def:closed-clone}  recovers the usual $\beta\eta$-laws for exponentials in $\stlc$, complete with the weakenings that are usually implicit. 
Writing $\app{f}{\varx}$ for $\eval$, we get the following equations in the internal language when
	$\Ctx := (\varx_i : \type_i)_{i=1, \dots, n}$:
\[	
	\csubbig
			{( \app{f}{\varx}) }
			{ \csub
					{(\abst{\varx}{\mmap})}
					{\varx_1 / x_1, \dots, \varx_n / x_n} / f,
				\varx / x
			}
		= \mmap 
	\hspace{1mm}
	,
	\hspace{1mm}
	\abst
		{\varx}
		{ 
			\csubbig
					{( \app{f}{\varx}) }
					{\csub
							{\mmap}
							{\varx_1 / x_1, \dots, \varx_n / x_n} / f}
		}
		=
		\mmap
\]


\begin{lemma}
	\label{res:free-property-of-stlc}
	The composite forgetful functor
		$\ClClone \to \Clone \to \Sig$
	has a left adjoint, and the free closed clone on $\sig$ is the syntactic clone $\Syn{\stlc_\sig}$.
\end{lemma}
%

%% file: sec-cartesian-closed-clones.tex
The development above makes defining cartesian closed structure straightforward. For reasons of space we restrict ourselves to the cartesian case, but similar remarks apply to the linear and ordered cases.

\begin{definition}
	A \emph{cartesian closed clone} is a clone equipped with both closed structure and cartesian structure. We write $\CCClone$ for the category of cartesian closed clones and homomorphisms that strictly preserve both structures.
\end{definition}

By \Cref{res:free-property-of-stpc,res:free-property-of-stlc}, we already have a free property .

\begin{lemma}
	The composite forgetful functor
		$\CCClone \to \Clone \to \Sig$
	has a left adjoint, and 
		$\Syn{\stlpc_{\sig}}$	
	is the free cartesian closed clone on $\sig$.
\end{lemma}

The nucleus of any cartesian closed clone
	$\ccc\clone$
is also cartesian closed:
\[
	\nucleus{\clone}(\type \times \type[2], \type[3])
	=	
	\clone(\type \times \type[2]; \type[3])
	\iso 
	\clone(\type, \type[2]; \type[3])
	\iso 
	\clone(\type; \expObj{\type[2]}{\type[3]})
	=
	\nucleus{\clone}(\type, \expObj{\type[2]}{\type[3]})
\]
Similarly, by
	\Cref{ex:cartesian-category-to-cartesian-clone,ex:ccc-to-closed-clone}, 
for any cartesian closed category $\ccc\cat$ the induced category 
	$\catToClone\cat$
is cartesian closed.
\Cref{res:factoring-adjunction-for-cartesian-clones} then restricts as follows.

\begin{proposition}
	\label{res:factoring-adjunction-for-cc-clones}
	The functor 
		$\nucleus{(-)} : \CCClone \to \CCCat$
	fits into the following diagram, in which the right-hand adjunction is an equivalence:
	\[
		\begin{tikzcd}
			\Sig & \CCClone & \CCCat
			\arrow[""{name=0, anchor=center, inner sep=0}, "\free"{}, yshift=2mm, from=1-1, to=1-2]
			\arrow[""{name=1, anchor=center, inner sep=0}, "\forget"{}, yshift=-2mm, from=1-2, to=1-1]
			\arrow[""{name=2, anchor=center, inner sep=0}, "\nucleus{(-)}", yshift=2mm, from=1-2, to=1-3]
			\arrow[""{name=3, anchor=center, inner sep=0}, "\catToClone", yshift=-2mm, from=1-3, to=1-2]
			\arrow["\simeq\adjUp"{anchor=center, rotate=0}, draw=none, from=3, to=2]
			\arrow["\adjUp"{anchor=center, rotate=0}, draw=none, from=0, to=1]
		\end{tikzcd}
	\]
	Moreover,
			$\forget \circ \catToClone$
	is equal to the canonical forgetful functor $\CCCat \to \Sig$.
	Hence, the free cartesian closed category on $\sig$ is canonically isomorphic to 
		$\nucleus{\Syn{\stlpc_{\sig}}}$.
\end{proposition}

As in \Cref{sec:recovering-the-semantics-for-cartesian}, the preceding two results are enough to recover the sound semantic interpretation of $\stlpc$, and the usual syntactic model.

%% file: sec-combinatory-logic-and-SK-clones.tex
In this section we begin a multi-ary investigation of cartesian combinatory logic, and give a categorical statement of the classical correspondence between combinatory logic and $\stlc$
	(for which see \eg~\cite{Gilezan1993,Bimbo2012}).
In \Cref{sec:SK-categories} we shall use this
to define \emph{SK-categories} and show they are sound and complete for $\stlc$.
%
%

We briefly recapitulate the rules of typed combinatory logic $\CL_\sig$ over a signature $\sig$; for a fuller account see~\eg~\cite{Bimbo2012}.
%
%
Types are as in $\stlc$. Terms are given by the grammar
$
	\mmap, \mmap[2] 
		::= \varx 
			\suchthat c \in \sig(\Ctx; \type[2]) 
			\suchthat (\app{\mmap}{\mmap[2]}) 
			\suchthat \combS 
			\suchthat \combK$: 
we have variables, constants and an application operation as in $\stlc$ and, for any context $\Ctx$ and types $\type, \type[2]$ and $\type[3]$, two \emph{combinators} 
	$\Ctx \vdash \combS^{\Ctx}_{\type, \type[2], \type[3]} :
		\expObj
				{ 
				\big( \expObj
							{\type}
							{(\expObj{\type[2]}{\type[3]})} \big) 
				}
				{\big(
					\expObj
						{(\expObj{\type}{\type[3]})}
						{(\expObj{\type}{\type[3]})}
				\big)}$ 
and 
	$\Ctx \vdash 
		\combK^{\Ctx}_{\type, \type[2]} : \expObj
					{\type}
					{
						(\expObj
							{\type[2]}
							{\type})
					}$.
%
%
Substitution is as in $\stlc$, where the combinators 
	$\combZ \in \{ \combS, \combK \}$
satisfy 
	$\csub{\combZ}{\mmap[2]_1 / \varx_1, \dots, \mmap[2]_n / \varx_n} = \combZ$
so that 
	$\combZ^{\Ctx}$
is the weakening of $\combZ^{\diamond}$.
	
The correlate of $\beta$-equality is \emph{weak equality} $\weq$, which is the smallest congruence containing 
$
	\appfour{\combS}{\varx}{\varx[2]}{\varx[3]}
		= \app
				{  (\app{\varx}{\varx[3]} ) }
				{  (\app{\varx[2]}{\varx[3]} ) }$
and
	$\appthree{\combK}{\varx}{\varx[2]} = \varx$.
The correlate of $\beta\eta$-equality is \emph{extensional weak equality} $\wexteq$, which extends $\weq$ with the rule
\vspace{-1mm}
\begin{equation} \label{eq:extensionality}
	\binaryRule
		{  \app{\mmap}{\varx_1 \, \cdots \, \varx_n} = \app{\mmap'}{\varx_1 \, \cdots \, \varx_n}  }
		{  \varx_1, \dots, \varx_n \text{ not free in } \mmap \text{ or } \mmap' }
		{\mmap = \mmap'}
		{ext}
\end{equation}
%

We write 
	$\CLw$
for combinatory logic with weak equality and 
	$\CLwext$ 
for combinatory logic with extensional weak equality. 
The usual encoding of $\CLw$ in $\stlc$ sends $\combS$ and $\combK$ to 
	$\abst
		{f}
		{\abst{g}
			{\abst{x}
				{ 
					\app
						{ (\app{f}{x}) }
						{ (\app{g}{x}  )}
	 			}
	 		}
	 	}$
and 
	$\abst
		{x}
		{\abst{y}{x}}$,
respectively.

The next definition may be obtained by seeing that $\CLw$ can be presented as an algebraic theory, and that clones are equivalent to algebraic theories (\eg~\cite{Linton1966,Taylor1973}).
We implicitly bracket application to the left, so 
	$\clappthree{\mmap}{\mmap[2]}{\mmap[3]}
		:=  \clapp
					{( \clapp{\mmap}{\mmap[2]} )}
					{\mmap[3]}$.
We also write 
	$\wkn{(-)}{\Ctx[2]; \Ctx[3]}$
for the weakening map
	$\clone(\Ctx; \type[2]) \to \clone(\Ctx[2], \Ctx, \Ctx[3]; \type[2])$
sending 
	$\mmap$
to
	$\csubbig
			{ \mmap }
			{ \p{\length{\Ctx[2]} + 1} {\Ctx[2], \Ctx, \Ctx[3]},
				\dots,
				\p{\length{\Ctx[2]} + \length{\Ctx}} {\Ctx[2], \Ctx, \Ctx[3]},
			}$;
when $\Ctx$ is empty we write just $\wkn{(-)}{\Ctx[2]}$.

\begin{definition}
	An \emph{SK-clone} is a clone $\clone$ equipped with a mapping 
		$\lolliObj{-}{{=}} : \Obj\clone \times \Obj\clone \to \Obj\clone$,
	nullary multimaps 
			$\combS_{\type, \type[2], \type[3]} 
				\in \clone{\big(\diamond ; 
					\lolliObjbig
							{ 
							\lolliObj
										{\type}
										{\lolliObj{\type[2]}{\type[3]}} 
							}
							{
								\lolliObj
									{\lolliObj{\type}{\type[2]}}
									{\lolliObj{\type}{\type[3]}}
							}\big)}$
			and
				$\combK_{\type, \type[2]}
					\in \clone{\big(\diamond; 
								\lolliObj
									{\type}
									{\lolliObj{\type[2]}{\type}} \big)}$
			for every $\type, \type[2], \type[3] \in \Obj\clone$,
		and a binary \emph{application} operation 
			$(\clapp{-}{{=}}):
						\clone(\Ctx; \lolliObj{\type}{\type[2]})
												\times \clone(\Ctx; \type) 
										\to \clone(\Ctx; \type[2])$ 			
			for every 
				$\Ctx \in \Obj{\clone}^{\star}$ and $\type[2] \in \Obj\clone$,
		such that the following axioms hold whenever they are well-typed:
	\begin{align*}
			\csub
				{(\clapp{\mmap}{\mmap[2]})}						
				{\mmap[3]_1, \dots, \mmap[3]_n}
			=
			\clapp	
				{  \csub{\mmap}{\mmap[3]_1, \dots, \mmap[3]_n}  }
				{  \csub{\mmap[2]}{\mmap[3]_1, \dots, \mmap[3]_n} }
			\hspace{3mm}
			&,
			\hspace{3mm}
			\clappthree
				{  ({\combK}
						_{\type, \type[2]})
						^{\type, \type[2]}  }
				{\p{1}{}}
				{\p{2}{}} = \p{1}{} 			
			\\[2mm]
			\clappfour
				{
						({\combS}
						_{\type, \type[2], \type[3]})
						^{
							\lolliObj{\type}{\lolliObj{\type[2]}{\type[3]}}, 
							\lolliObj{\type}{\type[2]},
							\type
						}
				}
				{  \p{1}{}  }
				{  \p{2}{}  }
				{  \p{3}{}  }
				&= \clapp 
						{( \clapp{\p{1}{}}{\p{3}{}} )}
						{(  \clapp{\p{2}{}}{\p{3}{}}  )} 
	\end{align*}
	A \emph{homomorphism of SK-clones} is a clone homomorphism that preserves application, $\combS$ and $\combK$:
	$
		\cloneFunctor(\combS_{\type, \type[2], \type[3]})
			= \combS_{\cloneFunctor\type, \cloneFunctor\type[2], \cloneFunctor\type[3]}
		%
		,
		%
		\cloneFunctor(\combK_{\type, \type[2]})
			= \combK_{\cloneFunctor\type, \cloneFunctor\type[2]}$
		%
		and
		%
	$\cloneFunctor(\clapp{\mmap}{\mmap[2]})
			= \clapp 
					{\cloneFunctor\mmap}
					{\cloneFunctor\mmap[2]} 			
	$.
	We write $\SKClone$ for the category of SK-clones and their homomorphisms.
\end{definition}


\begin{lemma}
	\label{res:free-property-of-clweak}
	The composite forgetful functor
		$\SKClone \to \Clone \to \Sig$
	has a left adjoint, and the free SK-clone on $\sig$ is the syntactic clone $\Syn{\CLw_\sig}$.
\end{lemma}

A core feature of the syntax of combinatory logic, which is at the heart of the correspondence between the terms of $\CLwext$ and $\stlc$, is the admissibility of \emph{bracket extension} algorithms~(see~\eg~\cite[\S7.1]{Barendregt1985}). 
To express this in the typed setting, we use the following notation. For a binary operation $\lolliObj{-}{{=}}$ on a set $\sorts$ we define 
	$\lolliCtx{-}{{=}} : \sorts^\star \times \sorts \to \sorts$ 
inductively as follows:
\[
			\lolliCtx{\diamond}{\type[2]} := \type[2] 
			\quad 
			,
			\quad
			\lolliCtx{\type}{\type[2]} := \lolliObj{\type}{\type[2]}
			\quad 
			,
			\quad
			\lolliCtx{\Ctx, \type}{\type[2]} 
				:= \lolliCtx
						{\Ctx}
						{\lolliObj{\type}{\type[2]}}
\]
With this notation, bracket abstraction amounts to saying that if 
	$\Ctx := (\varx_i : \type_i)_{i=1, \dots, n}$ 
and
	$\Ctx \vdash \mmap : \type[2]$
in $\CLw$, there exists a closed term 
	$\diamond \vdash \close\mmap : \lolliCtx{\Ctx}{\type[2]}$
such that
	$\app
		{  \wkn{(\close\mmap)}{\Ctx}  }
		{\varx_1 \, \ldots \, \varx_n}
	\weq \mmap$.
The extensionality axiom~(\ref{eq:extensionality}) then says that $\close\mmap$ is unique: in other words, 
	$\mmap \mapsto  
		\app
			{  \wkn{\mmap}{\Ctx}  }
			{\varx_1 \, \ldots \, \varx_n}$ 
	is an isomorphism.

We now translate this into clone-theoretic terms. For any SK-clone $\clone$ we obtain the operation 
	$\mmap \mapsto  
		\app
			{  \wkn{\mmap}{\Ctx}  }
			{\varx_1 \, \ldots \, \varx_n}$
as the composite below:
\begin{equation}
	\label{eq:def-of-itapp}
	\itapp{\Ctx}{\type[2]}
	:=
	\Big(
	\clone(\diamond; \lolliCtx{\Ctx}{\type[2]})
		\xra{\w^{\Ctx}}
		\clone(\Ctx; \lolliCtx{\Ctx}{\type[2]})
		\xra{\clapp{(-)}{\p{1}{\Ctx} \cdot \dots \p{\length\Ctx}{\Ctx}}}
		\clone(\Ctx; \type[2])
	\Big)
\end{equation}
For $\Ctx := \diamond$ this is just the identity.
The admissibility of bracket abstraction in the syntax of $\CLw$ is then captured by the next lemma. Typically bracket abstraction algorithms restrict to closed constants, because an open constant may have no corresponding closed term. We restrict in the same way. Call a signature $\sig$ \emph{nullary} if $\sig(\Ctx; \type) = \emptyset$ whenever $\Ctx \neq \diamond$, and write $\Sig_0 \hookrightarrow \Sig$ for the full subcategory of nullary signatures.

\begin{lemma}
	\label{res:weakening-and-application-in-SK-clones}
	Let $\sig$ be a nullary signature. 
	Then for any 
		$\Ctx \in \Obj{\Syn{\CLw_\sig}}^{\star}$ 
	and 
		$\type[2] \in \Obj{\Syn{\CLw_\sig}}$
	there exists a map 
		$\close{(-)}$
	such that 
		$\itapp{\Ctx}{\type[2]} \circ \close{(-)} =  \id_{\Syn{\CLw_\sig}}$.
\end{lemma}

Because bracket abstraction is defined by induction on the syntax, we cannot straightforwardly define it in an arbitrary SK-clone. We can, however, consider the sub-category of SK-clones (= semantic models of $\CLw$) which admit bracket abstraction in the sense that each $\itapp{\Ctx}{\type[2]}$ has a retraction.
The \emph{extensional} models are then those for which this retract $\close{(-)}$ also satisfies uniqueness.

\begin{definition}
	An SK-clone $\clone$ is \emph{extensional} if for every 
		$\Ctx \in \Obj{\clone}^{\star}$
	and 
		$\type[2] \in \Obj\clone$
	the map 
		$\itapp{\Ctx}{\type[2]}$
	defined in~(\ref{eq:def-of-itapp})  is invertible.
	We write 
		$\ExtSKClone$
	for the full subcategory of $\SKClone$ consisting of just the extensional SK-clones.
\end{definition}

\begin{lemma}
	\label{res:free-property-of-clext}
	The composite forgetful functor
		$\ExtSKClone \to \Clone \to \Sig_0$
	has a left adjoint, and the free extensional SK-clone on a nullary signature $\sig$ is the syntactic clone $\Syn{\CLwext_\sig}$.
\end{lemma}

\subsection{Extensional SK-clones are closed clones}
\label{sec:closed-clones-as-SK-clones}

In this section we outline why $\ExtSKClone$ is equivalent to $\ClClone$, thereby giving a category-theoretic equivalence not just between the syntax of $\CLwext$ and $\stlc$ but also between their models. 
The proof uses extensionality or the $\eta$-law to pass from arbitrary multimaps to nullary ones, from which one can build a \emph{strict closed} clone. 
We shall rely heavily on the following simple observation.

\begin{restatable}{lemma}{IsomorphismOfClonesFromIsoOfPresheaves}
	\label{res:isomorphism-of-clones-from-isomorphism-of-presheaves}
	Let $\clone$ be a clone and 
		$\objX
		:=
		\big\{ 
			\objX(\Ctx; \type[2])
		\big\}_{\Ctx \in \Obj{\clone}^\star, \type[2] \in \Obj\clone}$
	a family of sets together with an isomorphism
		$\big\{ \nu_{\Ctx; \type} : \clone(\Ctx; \type) \to \objX(\Ctx; \type) \big\}_{\Ctx, \type}$ 
	between $\objX$ and the hom-sets of $\clone$ in the functor category
		$\big[ \Obj{\clone}^{\star} \times \Obj{\clone}, \Set \big]$.
	Then $\objX$ acquires a canonical clone structure and
		$\nu$ 
	becomes an isomorphism of clones.
\end{restatable}

We now introduce strict closed clones.

\begin{definition}
	A \emph{strict closed clone} is a closed clone $\closed\clone$ such that every 
		$\Lambda^{\type} 
			: 
			\clone(\Ctx, \type; \type[2]) 
			\to 
			\clone(\Ctx, \expObj{\type}{\type[2]})$
	is the identity.
	We write 
		$\inc : \StClClone \hookrightarrow \ClClone$
	for the full subcategory  consisting of just the strict closed clones. 
\end{definition}


Any closed clone $\closed\clone$ determines a strict closed clone $\toStForClone\clone$ and a clone isomorphism 
	$\lambda_{\clone} : \clone \to \toStForClone\clone$
by applying 
	\Cref{res:isomorphism-of-clones-from-isomorphism-of-presheaves}
to the isomorphisms
	$\clone(\Ctx; \type[2]) \iso \clone(\diamond; \expCtx{\Ctx}{\type[2]})$
arising from the closed structure.
This extends to a functor 
	$\toStForClone : \ClClone \to \StClClone$
sending 
	$\cloneFunctor  : \closed\clone \to \closed{\clone[2]}$
to the composite
	$\lambda_{\clone[2]} \circ \cloneFunctor \circ \lambda_{\clone}^{-1}$.
A short calculation shows that the isomorphisms $\lambda$ make
	$\toStForClone : \ClClone \leftrightarrows \StClClone : \inc$
into an equivalence of categories.

We play a similar game for turning extensional SK-clones into (strict) closed clones. Indeed,  for any extensional SK-clone we have isomorphisms 
	$\clone(\Ctx; \type[2]) \iso \clone(\diamond; \lolliCtx{\Ctx}{\type[2]})$
defining a strict closed clone 
	$\toStForCL\clone$
with 
	$(\toStForCL\clone)(\Ctx; \type[2]) 
		:= \clone(\diamond; \lolliCtx{\Ctx}{\type[2]})$,
and hence a functor
	$\toStForCL : \ExtSKClone \to \StClClone$
in a similar fashion to $\toStForClone$.

Finally, for any closed clone $\closed\clone$ we get an extensional SK-clone 
	$\toCl\clone$
with the same underlying clone by taking application to be application in $\stlc$, so 
	$\clapp{\mmap}{\mmap[2]}
		:= \csub{\eval_{\type, \type[2]}}{\mmap, \mmap[2]}$,
and encoding the combinators as usual.
%
%



\begin{restatable}{theorem}{MainTheoremForClones}
	\label{res:main-theorem-for-clones}
	There exist equivalences of categories 
	\vspace{0mm}
	\begin{td}[ampersand replacement = \&]
		\ExtSKClone 
			\arrow[yshift=2mm]{r}{\toStForCL}
			\arrow[phantom]{r}[font=\scriptsize, description]{\simeq} \&
		\StClClone 
			\arrow[yshift=-2mm]{l}{\toCl' := \toCl \circ \inc}
			\arrow[hookrightarrow, yshift=2mm]{r}{\inc} 
			\arrow[phantom]{r}[font=\scriptsize, description]{\simeq} \&
		\ClClone. 
			\arrow[yshift=-2mm]{l}{\toStForClone} 
	\end{td}
\end{restatable}

%% file: sec-SK-categories.tex
In 
	\Cref{res:factoring-adjunction-for-cartesian-clones,res:factoring-adjunction-for-cc-clones}
we recovered a unary semantic interpretation of $\stpc$ and $\stlpc$ from our clone-theoretic ones. 
But we do not have a corresponding result for $\stlc$.
In this section we fill this gap: we introduce \emph{SK-categories} and show they play the role for $\stlc$ that cartesian closed categories play for $\stlpc$.
Our definition is inspired by
	\emph{closed categories}~(\cite{Eilenberg1966,Day1978}),
%
%
%
which axiomatise an `internal' version of the hom-functor 
	$\cat(-, {=})$
in the form of a functor
	$\lolliObj{-}{{=}} : \op\cat \times \cat \to \cat$.
%
%
Closed categories have a unit object, corresponding to requiring a unit type~(\cf~\cite{Manzyuk2012}); our definition avoids this
	(see also~\cite{ShulmanNlab,Uustalu2020}).
%

%
Recall that in the presence of contravariance, \emph{dinaturality} and \emph{extranaturality} are the right replacements for naturality
	(see~\eg~\cite[\S IX.4]{cfwm}).
	
\begin{definition}
	\label{def:SK-category}
	An \emph{SK-category} consists of a category $\cat$ and functors
		$\lolliObj{-}{{=}} : \op\cat \times \cat \to \cat$ 
	and 
		$\skforget : \cat \to \Set$,
	together with 
	\begin{enumerate}
	%
	\item 
		Maps
		$\SS{\obj, \obj[2], \obj[3]}
						: 
						\lolliObj
							{\obj}
							{\lolliObj{\obj[2]}{\obj[3]}}
						\to 
						\lolliObj
							{\lolliObj{\obj}{\obj[2]}}
							{\lolliObj{\obj}{\obj[3]}}$
		dinatural in $\obj$ and natural  in $\obj[2]$ and $ \obj[3]$;
	\item 
		Maps
			$\K{\obj}{\obj[2]}  : \obj[2] \to \lolliObj{\obj}{\obj[2]}$
		extranatural in $\obj$ and natural in $\obj[2]$;
	\item 
		Maps 
			$\catapp{\obj, \obj[2]} 
				: 
				\skforget\lolliObj{\obj}{\obj[2]}
					\times 
					\forget\obj
				\to 
				\skforget\obj[2]$
	extranatural in $\obj$ and natural in $\obj[2]$;
	\end{enumerate}
	This data is subject to the condition that 
		$\skforget \circ \lolliObj{-}{{=}} = \cat(-, {=})
							: \op\cat \times \cat \to \Set$
	and the 7 axioms of \Cref{fig:SK-category-axioms}.
	An \emph{SK-functor} 
		$(\functor, \homs, \forgs)$ 
	is a functor 
		$\functor : \cat \to \cat[2]$
	with natural transformations 
	as below, such that the axioms of 
			\Cref{fig:SK-functor-axioms} 
	hold. 
	\[
		\begin{tikzcd}[column sep = 3em]
			\op\cat \times \cat &
			\op{\cat[2]} \times \cat 
			\\
			\cat 
				\arrow{r}[swap]{\functor} &
			\cat[2]
			\arrow["\op\functor \times \functor", from=1-1, to=1-2]
			\arrow["\lolliObj{-}{{=}}"{name=0}, from=1-2, to=2-2]
			\arrow["\lolliObj{-}{{=}}"{name=1, swap}, from=1-1, to=2-1]
			\arrow[from=2-1, to=2-2]
			\arrow["\phi", shorten <=18pt, shorten >=18pt, Rightarrow, from=1, to=0]
		\end{tikzcd}
		\hspace{1.8cm}
		%
		\begin{tikzcd}[]
			\cat && {\cat[2]} \\
			& \Set
			\arrow["\functor", from=1-1, to=1-3]
			\arrow[""{name=0, anchor=center, inner sep=0}, "\skforget", from=1-3, to=2-2]
			\arrow[""{name=1, anchor=center, inner sep=0}, "\skforget"', from=1-1, to=2-2]
			\arrow["\forgs", shorten <=15pt, shorten >=15pt, Rightarrow, from=1, to=0]
		\end{tikzcd}
	\]
	We call $\skfunctor{\functor}$ \emph{strict} if $\phi$ is the identity, and write $\StSKCat$ for the category of SK-categories and strict SK-functors.
\end{definition}

\begin{figure}
\centering 
%
\begin{subfigure}{\textwidth}
\vspace{-2mm}
\input{fig-closed-category-axioms}
\vspace{0mm}
\caption{
	Axioms for an SK-category. 
	In (1) the unlabelled arrow is the canonical map
		$\seqlr{ \seq{ \pi_1 \pi_1, \pi_2 }, \seq{  \pi_2 \pi_1, \pi_2 } }
			:
			(\objX \times \objX[2]) \times \objX[3] 
			\to 
			(\objX \times \objX[3]) \times (\objX \times \objX[3])$.
	In (3) we write $\ulcorner \id_C \urcorner$ for the set map $\ast \mapsto \id_C : 1 \to \forget[C, C]$. 	
	}
\vspace{-2mm}
\label{fig:SK-category-axioms}
\end{subfigure}
\noindent\hspace{0cm}\makebox[10cm]{\rule{\textwidth}{0.4pt}}
\vspace{3mm}
\begin{subfigure}{\textwidth}
\input{fig-SK-functor-axioms}
\vspace{-1mm}
\caption{Axioms for an SK-functor}
\vspace{-5mm}
\label{fig:SK-functor-axioms}
\end{subfigure}
\noindent\hspace{0cm}\makebox[10cm]{\rule{\textwidth}{0.4pt}}

\vspace{-2mm}
\caption{Extra axioms for \Cref{def:SK-category}}
\end{figure}	

We think of $\skforget\obj$ as the set of multimaps $\diamond \to \obj$ and $\catapp{}{}$ as a formal application operation $(\clapp{-}{{=}})$. Axioms (1) and (2) are the weak equality laws from $\CL$. 
Axioms (3) and (4) ensure compatibility between the category structure and the corresponding $\CL$ constructions: for example, 
axiom (3) implies
	$\skforget(f)(\varx) = \clapp{f}{\varx}$, 
and  axiom (4) says that composition coincides with
	$\appthree{\combS}{(\app{\combK}{-})}{({=})}$, 
corresponding to the weak equality
	$\appfour
		{\combS}
		{(\app{\combK}{f})}
		{g}
		{\varx}
	=
		\app
			{f}
			{(\app{g}{\varx})}$.
Axioms (5) -- (7) are coherence laws.

%

Every extensional SK-clone determines an SK-category. 
Because we follow~\cite{Eilenberg1966} and ask for an \emph{equality}
	$\skforget\lolliObj{\type}{\type[2]} = \cat(\type, \type[2])$
in the definition of SK-categories, but in general an extensional SK-clone 
	$\sk\clone$
only has an \emph{isomorphism} 
	$\clone(\type; \type[2])
		\iso \clone(\diamond; \lolliObj{\type}{\type[2]})$,
we need to strictify in the same manner as 	
	\Cref{sec:closed-clones-as-SK-clones}.
As a notational shorthand, we write 
	$\combI, \combB$ 
and
	$\combBprime$
for the closed multimaps satisfying the equations below in the internal language of $\clone$
 	(see~\eg~\cite{Gilezan1993,Bimbo2012}):
\[
	\clapp
				{\wkn{\combI}{\type}}
				{\varx} 
			= \varx
	\hspace{1mm}
	,
	\hspace{1mm}
	\clappfour
			{  \combB
				%
					^{\expObj{\type[2]}{\type[3]}, \expObj{\type}{\type[2]}, \type}  
			}
			{\varx}{\varx[2]}{\varx[3]}
		= \clapp{\varx}{(\clapp{\varx[2]}{\varx[3]})}	
	\hspace{1mm}
	,
	\hspace{1mm}
	\clappfour
			{  (\combB')
					^{\expObj{\type}{\type[2]}, \expObj{\type[2]}{\type[3]}, \type}  
			}
			{\varx}{\varx[2]}{\varx[3]}
					= \clapp
							{\varx[2]}{(\clapp{\varx}{\varx[3]})}
\]
	
The category 
	$\toSKCat\clone$
has objects $\Obj\clone$ and hom-sets 
	$(\toSKCat\clone)(\type, \type[2]) 
		:= \clone(\diamond; \lolliObj{\type}{\type[2]})$ (\cf~\cite{Fox1971}).
The identity on $\type$ is $\combI_{\type}$ and the composite of
	$\mmap$ 
and  
	$\mmap'$ 
is
	$\clappthree
			{\combB} 
			{\mmap}
			{\mmap'}$.
%
%
%
For $\skforget$ we take 
	$\skforget\type := \clone(\diamond; \type)$
with the action on maps given by application. 
%
%
%
For 	
	$\lolliObj{-}{{=}}$
the action on objects is given by the SK-structure, with the action on maps given by
$
	\lolliObj{\objX}{\mmap} := 
		\clapp
			{\combB} 
			{\mmap}
$
and
$
	\lolliObj{\mmap}{\objX} := 
		\clapp
			{\combB'} 
			{\mmap}	
$.
The maps
	$\SS{}$ and $\K{}{}$ 
are given by the corresponding combinators, and $\catapp{}$ is the application operation in $\clone$.
This extends to a functor 
	$\toSKCat : \ExtSKClone \to \StSKCat$.

The internal language of SK-categories is $\CLwext$, and hence $\stlc$. 
%
%
We write 
	$\forget$
for the functor which sends an SK-category 
	$\skcat\cat$
to the signature with base types $\Obj\cat$ and constants
	$\skforget\lolliObj{\Ctx}{\type[2]}$. 
	
\begin{proposition}
	\label{res:internal-language-of-SK-cats}
	The forgetful functor 
		$\forget : \StSKCat \to \Sig$
	has a left adjoint, and the free SK-category on $\sig$ is 
		$
			\toSKCat\big( \Syn{\CLwext_\sig} \big)
				\iso (\toSKCat \circ \toCl)\big( \Syn{\stlc_\sig} \big)
		$. 
\end{proposition} 

Using \Cref{res:main-theorem-for-clones}, we now obtain a version of 
		\Cref{res:factoring-adjunction-for-cartesian-clones,res:factoring-adjunction-for-cc-clones}
for $\stlc$.

\begin{restatable}{theorem}{SKCategoriesAreClosedClones}
	\label{res:factoring-adjunction-for-closed-clones}
	The composite 
		$\toSKCat \circ \inc : \StClClone \to \SKCat$
	is invertible; hence we get the diagram  below, in which the right-hand adjunction is an equivalence:
	\[
		%
		\hspace{1cm}
		%
		\begin{tikzcd}[ampersand replacement = \&]
			\Sig \& \ClClone \& \SKCat
			\arrow[""{name=0, anchor=center, inner sep=0}, "\free"{}, yshift=2mm, from=1-1, to=1-2]
			\arrow[""{name=1, anchor=center, inner sep=0}, "\forget"{}, yshift=-2mm, from=1-2, to=1-1]
			\arrow[""{name=2, anchor=center, inner sep=0}, "\toSKCat \circ \toCl", yshift=2mm, from=1-2, to=1-3]
			\arrow[""{name=3, anchor=center, inner sep=0}, "\SKCatToClone", yshift=-2mm, from=1-3, to=1-2]
			\arrow["\simeq\adjUp"{anchor=center, rotate=0}, draw=none, from=3, to=2]
			\arrow["\adjUp"{anchor=center, rotate=0}, draw=none, from=0, to=1]
		\end{tikzcd}
	\]
	Moreover,
			$\forget \circ \SKCatToClone$
	is equal to the forgetful functor $\SKCat \to \Sig$, so
	 the free SK-category on $\sig$ is canonically isomorphic to 
		$(\toSKCat \circ \toCl)(\Syn{\stlc_{\sig}})$.
\end{restatable}

Recall that a \emph{closed monoidal category} is a monoidal category 
	$\monoidal{\cat[2]}$
such that every 
	${(-) \tens \obj[2]}$
has a right adjoint $\lolliObj{\obj[2]}{-}$, and that in a closed category $\cat$ giving every
	$\lolliObj{\obj}{-}$
a $\cat$-enriched left adjoint is equivalent to giving closed monoidal structure~(\cite{Eilenberg1966,Day1978,Uustalu2020}).
	\Cref{res:factoring-adjunction-for-closed-clones} 
and 
	\Cref{res:factoring-adjunction-for-cc-clones}
imply a cartesian~version.

\begin{restatable}{corollary}{CartesianClosedFromSK}
	Equipping a category
		$\cat$ 
	with cartesian closed structure is equivalent to equipping $\cat$ with SK-structure and natural isomorphisms 
		$\cat(\tensu, \lolliObj{\obj}{\obj[2]})
				\iso
				\cat(\obj, \obj[2])$
	and
		$\cat(\obj \tens \obj[2], \obj[3])
			\iso 
			\cat(\obj, \lolliObj{\obj[2]}{\obj[3]})$	
	for every 
		$\obj, \obj[2], \obj[3] \in \cat$.
\end{restatable}

%% file: fig-closed-category-axioms.tex
\[
	\hspace{0mm}
	(1) \hspace{-3mm}
	\begin{tikzcd}[scalenodes = .95, column sep = -6em, row sep = 2em]
		\:
		& 
				\big(
					\skforget
							\lolliObj
								{\obj}
								{\lolliObj{\obj[2]}{\obj[3]}}  
				\times 
					\skforget\lolliObj{\obj}{\obj[2]} \big)
				\times \skforget\obj
			\\
		\node[align = center, name = 0]
			{$
			\big( 
						\skforget
							\lolliObj
								{\obj}
								{\lolliObj{\obj[2]}{\obj[3]}}
						\times 
							\skforget\obj
					\big)
				$ \\[1mm] 
				\hspace{5mm}$
				\times \big(
					\skforget\lolliObj{\obj}{\obj[2]} \big)
						\times 
						\skforget\obj
				\big)
				$
			};
		\: &
		\: & 
		\node[align = center, name = 1]
			{$
				\skforget
					\lolliObj
						{\lolliObj{\obj}{\obj[2]}}
						{\lolliObj{\obj}{\obj[3]}} 
				$ \\[1mm]\hspace{12mm}$			
				\times 
					\skforget\lolliObj{\obj}{\obj[2]} 
			\times \skforget\obj
		$	};
		\\
		\skforget\lolliObj{\obj[2]}{\obj[3]}
			\times \skforget\obj[2] & 
		\: & 
		\skforget\lolliObj{\obj}{\obj[3]}
			\times \skforget\obj 
		\\
		\: &
		\skforget\obj[3] &
		\:
		\arrow[from=1-2, to=0]
		\arrow["\catapp{} \times \catapp{}"{swap}, from=0, to=3-1]
		\arrow["\catapp{}"{swap}, from=3-1, to=4-2]
		\arrow["\skforget\SS{} \times \id \times \id"{yshift=-1mm, xshift=.5mm}, from=1-2, to=1]
		\arrow["\catapp{} \times \id", from=1, to=3-3]
		\arrow["\catapp{}", from=3-3, to=4-2]
	\end{tikzcd}
	\hspace{-11mm}
	\begin{minipage}{0.6\textwidth}
	\centering
	\[
	(2) \hspace{-1mm}
	\begin{tikzcd}[scalenodes = .9, column sep = 1em]
		{\skforget\obj[2] \times \skforget\obj} & {\skforget\lolliObj{\obj}{\obj[2]} \times \skforget\obj} \\
		& {\skforget\obj[2]}
		\arrow["{\skforget\K{}{} \times \id}"{yshift=1mm}, from=1-1, to=1-2]
		\arrow["{\catapp{}}", from=1-2, to=2-2]
		\arrow["{\pi_1}"', from=1-1, to=2-2]
	\end{tikzcd}
	\]
	%
	%
	%
	%
	\[
	(3) \hspace{-1mm}	
	\begin{tikzcd}
		\terminal \times \skforget\obj
		\arrow{r}[yshift = 1mm]{ \nameOf{\id_{\obj}} \times \id }
		\arrow{dr}[swap]{\pi_2}
		&
		\skforget\lolliObj{\obj}{\obj}
			\times
			\skforget\obj
		\arrow{d}{  \catapp{}  }
		\\
		\: 
		&
		\skforget\obj
	\end{tikzcd}
	\]
	\end{minipage}	
\]
\vspace{0mm}
%
%
	\[
	\centering
	\hspace{-2mm}
	(4) \hspace{1mm}
	\begin{tikzcd}[column sep = 4em, scalenodes = 1]
			\skforget\lolliObj{\obj[2]}{\obj[3]}
				\times 
				\skforget\lolliObj{\obj}{\obj[2]}
			\arrow{r}{\circ}
			\arrow{d}[swap]{  \skforget\K{}{} \times \id  }
			&
			\skforget\lolliObj{\obj}{\obj[3]}
			\\
			\skforget
				\lolliObj
					{\obj}
					{\lolliObj{\obj[2]}{\obj[3]}}
			\times 
			\skforget\lolliObj{\obj}{\obj[2]}		
			\arrow{r}[swap]
				{\skforget\SS{} \times \id}		
			&
			\skforget
				\lolliObj
					{  \lolliObj{\obj}{\obj[2]}  }
					{  \lolliObj{\obj}{\obj[3]}  }
			\times 
			\skforget\lolliObj{\obj}{\obj[2]}		
			\arrow{u}[swap]{  \catapp{}  }			
	\end{tikzcd}		
	\]
%
%
\vspace{0mm}
\[
	\hspace{-5mm}
	(5)
	\hspace{-1mm}
	\begin{tikzcd}[column sep = -4.5em, scalenodes = .9]
		\: &
		{\left[ [X, A], \left[ [X, B], [X, C] \right]  \right]} &
		\: 
		\\
		{\big[ [X, A], [X, [B,C]] \big]} & \: & {\big[ \left[ [X, A], [X, B] \right] \left[ [X, A], [X, C] \right]  \big]} \\
		{\big[ X, [A, [B, C]] \big] } && {\big[ \left[ X, [A, B] \right], \left[ [X, A], [X, C] \right] \big]} \\
		{\big[ X, \left[ [A, B], [A, C] \right] \big]} && {\big[ \left[ X, [A, B] \right], \left[ X, [A, C] \right] \big]}
		\arrow["{\SS{}}", from=3-1, to=2-1]
		\arrow["{\lolliObj{\id}{\SS{}}}"{yshift=-1mm, xshift=-.5mm}, from=2-1, to=1-2]
		\arrow["{\SS{}}", from=1-2, to=2-3]
		\arrow["{\lolliObj{\SS{}}{\id}}", from=2-3, to=3-3]
		\arrow["{\lolliObj{\id}{\SS{}}}"', from=3-1, to=4-1]
		\arrow["{\SS{}}"', from=4-1, to=4-3]
		\arrow["{\lolliObj{\id}{\SS{}}}"', from=4-3, to=3-3]
	\end{tikzcd}
	\hspace{-11mm}
	\begin{minipage}{0.6\textwidth}
	\centering
	\[
	(6) \hspace{-2mm}
	\begin{tikzcd}[scalenodes = 1, column sep = 1.5em]
		\lolliObj{\obj}{\obj[3]}
		\arrow{r}{ \lolliObj{\id}{\K{\obj[2]}{}} }
		\arrow[swap]{d}{ \K{\obj[2]}{} }
		&
		\lolliObj
			{\obj}
			{\lolliObj{\obj[2]}{\obj[3]}}
		\\
		\lolliObj
			{\obj[2]}
			{\lolliObj{\obj}{\obj[3]}}
		\arrow{r}[swap]{  \SS{} }
		&
		\lolliObj
			{  \lolliObj{\obj[2]}{\obj}  }
			{  \lolliObj{\obj[2]}{\obj[3]}  }		
		\arrow{u}[swap]{ \lolliObj{\K{\obj[2]}{}}{\id} }
	\end{tikzcd}		
	\]
	\[
	(7) \hspace{-2mm}
	\begin{tikzcd}
		\lolliObj{\obj}{\obj[3]}
		\arrow{r}{  \lolliObj{\id}{\K{\obj[2]}{}}  }
		\arrow{dr}[swap]{ \K{\lolliObj{\obj}{\obj[2]}}{} }
		&
		\lolliObj
			{\obj}
			{  \lolliObj{\obj[2]}{\obj[3]}  }
		\arrow{d}{  \SS{}  }
		\\
		\: 
		&
		\lolliObj 
			{  \lolliObj{\obj}{\obj[2]}  }
			{  \lolliObj{\obj}{\obj[3]}	}
	\end{tikzcd}	
	\]		
	\end{minipage}
\]

%% file: fig-SK-functor-axioms.tex
\[
	\begin{tikzcd}
		\skforget^{\cat}\lolliObj{\obj}{\obj[2]} & \cat(\obj, \obj[2]) & \cat[2](\functor\obj, \functor\obj[2]) \\
		\skforget^{\cat[2]}\functor\lolliObj{\obj}{\obj[2]}  && \skforget^{\cat[2]}\lolliObj{\functor\obj}{\functor\obj[2]}
		\arrow["\forgs"', from=1-1, to=2-1]
		\arrow["\skforget\homs"', from=2-1, to=2-3]
		\arrow[equals, from=1-3, to=2-3]
		\arrow[equals, from=1-1, to=1-2]
		\arrow["{\functor_{\obj, \obj[2]}}", from=1-2, to=1-3]
	\end{tikzcd}
	\hspace{8mm}
	%
	\begin{tikzcd}
		{\functor\obj[2]} & {\functor\lolliObj{\obj}{\obj[2]}} \\
		& {\lolliObj{\functor\obj}{\functor\obj[2]}}
		\arrow["{\functor\K{\obj}{}}", from=1-1, to=1-2]
		\arrow["\homs", from=1-2, to=2-2]
		\arrow["{\K{\functor\obj}{}}"', from=1-1, to=2-2]
	\end{tikzcd}
\]
%
%
%
%
\[
\hspace{-4mm}
\begin{minipage}{0.5\textwidth}
\begin{tikzcd}[column sep = -2.2em, scalenodes = .9]
	\skforget^{\cat}\lolliObj{\obj}{\obj[2]} \times \skforget^{\cat}\obj &
	\: & 
	\skforget^{\cat}\obj[2] \\
	\skforget^{\cat[2]}\functor\lolliObj{\obj}{\obj[2]} \times \skforget^{\cat[2]}\functor\obj & 
	\: & 
	\skforget^{\cat[2]}\functor\obj[2] \\
	\: &
	\skforget^{\cat[2]}\lolliObj{{\functor}\obj}{{\functor}\obj[2]} \times \skforget^{\cat[2]}\functor\obj &
	\:
	\arrow["{\forgs \times \forgs}"', from=1-1, to=2-1]
	\arrow["{\skforget^{\cat[2]}\homs \times \id}"'{yshift=1mm}, from=2-1, to=3-2]
	\arrow["{\catapp{}^{\cat[2]}}"'{yshift=1mm}, from=3-2, to=2-3]
	\arrow["{\catapp{}^{\cat}}", from=1-1, to=1-3]
	\arrow["\forgs", from=1-3, to=2-3]
\end{tikzcd}
\vspace{8mm}
\end{minipage}
%
%
\hspace{0mm}
%
\begin{tikzcd}[column sep = -2.4em, scalenodes = .9]
	{\functor\big[ \obj, [\obj[2], \obj[3]] \big]} &
	\: & 
	\functor\big[  [\obj, \obj[2]], [\obj, \obj[3]] \big]  \\
	{\big[{\functor} \obj, {\functor}[\obj[2], \obj[3]] \big]} && \big[  {\functor}[\obj, \obj[2]], {\functor}[\obj, \obj[3]] \big] \\
	{\big[{\functor} \obj, [{\functor}\obj[2], {\functor}\obj[3]] \big]}  & 
	\:   & 
	\big[  {\functor}[\obj, \obj[2]], [{\functor}\obj, {\functor}\obj[3]] \big] \\
	\: &
	{\big[  [{\functor}\obj, {\functor}\obj[2]], [{\functor}\obj, {\functor}\obj[3]] \big]} &
	\: 
	\arrow["\homs"', from=1-1, to=2-1]
	\arrow["{\lolliObj{\id}{\homs}}"', from=2-1, to=3-1]
	\arrow["{\functor\SS{}}", from=1-1, to=1-3]
	\arrow["\homs", from=1-3, to=2-3]
	\arrow["{\lolliObj{\id}{\homs}}", from=2-3, to=3-3]
	\arrow["{\SS{}}"'{yshift=1mm, xshift=-1mm},, from=3-1, to=4-2]
	\arrow["{\lolliObj{\homs}{\id}}"'{yshift=1mm, xshift = 1mm}, from=4-2, to=3-3]
\end{tikzcd}
\]

%% file: sec-appendix.tex
We sketch some of the proofs omitted from the main paper.

\subsection{ \texorpdfstring{\Cref{sec:clones-and-multicats}}{Section 2} }


\SemanticInterpOfBasicLin*
\begin{proof}
The unique extension 
	$\interp\sem{-} : \Syn{\basicMon_{\sig}} \to \multicat$
of
	$\interp : \sig \to \forget\multicat$
is the semantic interpretation: 
\[
	\interp\sem{\varx} 
		= \Id_{\sem\type}
	\qquad
	,
	\qquad
	\interp\sem{
		c^{\S}(\mmap[2]_1, \dots, \mmap[2]_n)
	}
	= \msub
			{(\interp c)}
			{\interp\sem{\mmap[2]_1}, \dots, \interp\sem{\mmap[2]_n}}
\]
\qed
\end{proof}


\CloneToMulticat*
\begin{proof}
	The identity  $\Id_{\type}$ in $\toMulti\clone$ is $\p{1}{\type}$.
	Substitution is as follows.
	For multimaps
		$\mmap : \type_1, \dots, \type_n \to \type[2]$
	and
		$( \mmap[2]_i : \Ctx[2]_i \to \type_i)_{i = 1, \dots, n}$,
	we set 
		$\msub
			{\mmap}
			{ \mmap[2]_1, \dots,  \mmap[2]_n}
			: \Ctx[2]_1, \dots, \Ctx[2]_n \to \type[2]$
	to be the substitution 
		$\csubbig
			{\mmap}
			{ 
				\csub{\mmap[2]_1}{\p{\Ctx[2]_1}{}}, 
				\dots,  
				\csub{\mmap[2]_n}{{\p{\Ctx[2]_n}{}}}
			}$,
	where 
		$\csub
			{\mmap[2]_i}
			{\p{\Ctx[2]_i}{}}$
	is the composite
		$
			\Ctx[2]_1, \dots, \Ctx[2]_n \to \Ctx[2]_i \xra{\mmap[2]_i} \type_i
		$,
	in which the first arrow uses the obvious projections
		$\p{\sum_{j=1}^{i-1} \length{\Ctx[2]_j}}{},
		\dots,
		\p{\sum_{j=1}^{i} \length{\Ctx[2]_j}}{}$.
	\qed
\end{proof}


\subsection{ \texorpdfstring{\Cref{sec:cartesian-clones}}{Section 4}  }

\FactoringAdjunctionForCartesianClones*
\begin{proof}
	Set
		$\eta_{\clone} : \clone \to \catToClone(\nucleus\clone)$
	to be the identity-on-objects cartesian clone map which sends 
		$\mmap : \type_1, \dots, \type_n \to \type[2]$
	to 
		$\csub{\mmap}{\pi_1^{\ind\type}, \dots, \pi_n^{\ind\type}} : \prod_{i=1}^n \type_i \to \type[2]$.
	Now suppose 
		$\cloneFunctor : \clone \to \catToClone{\cat[2]}$
	is a clone homomorphism.
	We show there is a unique strict product-preserving functor 
		$\ext{\cloneFunctor} : \nucleus\clone \to \cat[2]$
	satisfying 
		$\catToClone(\ext{\cloneFunctor}) \circ \eta_{\clone} = \cloneFunctor$.
	Take $\ext{\cloneFunctor}$ to be the restriction of $\cloneFunctor$ to unary maps, so that 
		$(\catToClone(\ext{\cloneFunctor}) \circ \eta_{\clone})(\mmap)$
	is the following:
	\[	
	 	(\catToClone(\ext{\cloneFunctor})(\csub{t}{\pi_1^{\ind\type}, \dots, \pi_n^{\ind\type}})
	 	= \cloneFunctor(\csub{\mmap}{\pi_1^{\ind\type}, \dots, \pi_n^{\ind\type}})
	 	= (\cloneFunctor\mmap) \circ 
	 			\seq{\pi_1^{\ind\type}, \dots, \pi_n^{\ind\type}}
	 	= \cloneFunctor\mmap
	\]
	and by the unit law in $\catToClone\cat[2]$ this is exactly 
	$\cloneFunctor(\mmap)$, 
	as required.
	Uniqueness then follows by essentially the same argument.
	\qed
 \end{proof} 
 

For proving \Cref{res:representable-equals-cartesian}, we begin with the following lemma, which says that the universal arrows in a representable multicategory are automatically closed under composition.
	
	
\begin{lemma}
	\label{res:closure-of-universal-arrows-in-clones}
	For any clone $\clone$, the multicategory $\toMulti\clone$ is representable if and only if
	$\toMulti\clone$ is equipped with universal arrows 
		as in \Cref{def:representable-multicategory}
	for each 
		${\objX_1, \dots, \objX_n \in {\Obj\clone}^{\star}}$.
\end{lemma}
\begin{proof}
	The key observation is the following. Let us write 
		$\pi_i^{\ind\type} : \T_{i=1}^n \type_i \to \type_i$
	for $\ext{(\p{i}{\ind{\type}})}$, the unique map corresponding to 
		$\p{i}{\ind\type} : \type_1, \dots, \type_n \to \type_i$
	across the isomorphism given by representability.
	Then
	\begin{equation}
	\label{eq:properties-of-pi}
	\pi_i^{\ind{\type}} \circ \univ_{\ind\type} = \p{i}{\ind\type}			
	\qquad 
	,
	\qquad 
	\csub
		{\univ_{\ind\type}}	
		{\pi_1^{\ind\type}, \dots, \pi_n^{\ind\type}}
	=
		\id_{\T(\type_1, \dots, \type_n)}
	\end{equation}
	
	We now turn to proving the result. 
	We are given universal arrows as follows
	\begin{align*}
		\univ_{\ind\type} &: \type_1, \dots, \type_n \to \T_{i=1}^n \type_i \\
		\univ_{\ind{\type[2]}} &: \type[2]_1, \dots, \type[2]_m \to \T_{j=1}^m \type[2]_i \\
		\univ_{\T(\ind\type), \T(\ind{\type[2]})}
			&: \T_{i=1}^n \type_i, \T_{j=1}^m \type[2]_j 
			\to 
			\T{\big(  
					\T_{i=1}^n \type_i, \T_{j=1}^n \type[2]_m
				)\big)}
	\end{align*}
	and we need to show the composite 
		$\widetilde{\univ} :=
			\msub
				{\univ_{\T(\ind\type), \T(\ind{\type[2]})}}
				{\univ_{\ind\type}, \univ_{\ind{\type[2]}}}$
	in $\toMulti\clone$	is also universal.	
	To this end, set
	\[
		\objX_i :=
			\begin{cases}
				\type_i & \text{ for } i = 1, \dots, n \\
				\type[2]_{i - n} & \text{ for } i = n+1, \dots, n + m
			\end{cases}
	\]
	and define
		$\widetilde{\pi}_i : 	
				\T{\big(  
						{\T}_{i=1}^n \type_i, \T_{j=1}^n \type[2]_m
					)\big)} \to \objX_i$
	using iterated applications of the projections $\pi$:
	\[
		\widetilde{\pi}_i  	:=
			\begin{cases}
				\csub
					{\pi_i^{\ind\type}}
					{\pi_1^{\T(\ind\type), \T(\ind{\type[2]})}} & \text{ for } i = 1, \dots, n \\
				\csub
						{\pi_{i-n}^{\ind{\type[2]}}}
						{\pi_1^{\T(\ind\type), \T(\ind{\type[2]})}} & \text{ for } i = n+1, \dots, n + m
			\end{cases}
	\]
	Directly calculating using~(\ref{eq:properties-of-pi}), we see that
		$\p{i}{ \ind{\type}, \ind{\type[2]} }
					=
					\msub
						{\widetilde{\pi}_i}
						{\widetilde{\univ}}$
	and
		$\id_{\T(\T(\ind\type), \T(\ind{\type[2]}))}
					= 
					\csub 
						{\widetilde{\univ}}
						{\widetilde{\pi}_1, \dots, \widetilde{\pi}_{n+m}}$.
	Hence, setting
		$\ext{\mmap} := 
				\csub
					{\mmap}
					{\widetilde{\pi}_1, \dots, \widetilde{\pi}_{n+m}}$
	we get that $\ext{(-)}$ is inverse to 
		$
		\msub
			{(-)}
			{\msub
					{\univ_{\T(\ind\type), \T(\ind{\type[2]})}}
					{\univ_{\ind\type}, \univ_{\ind{\type[2]}}}}
		= 
		\msub
			{(-)}
			{\widetilde\univ}$, 
	as required.
	\qed
\end{proof}



\RepresentableEqualsCartesian*
\begin{proof}
		Suppose first that $\clone$ has representable structure. Then for  
			${\type_1, \dots, \type_n \in {\Obj\clone}^{\star}}$
		let 
			$\pi_i^{\ind\type} := \ext{(\p{i}{\ind\type})} : \T_{i=1}^n \type_i \to \type_i$
		be the multimap defined in 
			\Cref{res:closure-of-universal-arrows-in-clones}, 
		and for 
			$(\mmap_i : \Ctx \to \type_i)_{i=1, \dots, n}$
		set 
			$\seq{\mmap_1, \dots, \mmap_n}$
		to be 
			$\csub
				{\univ_{\ind\type}}
				{\mmap_1, \dots, \mmap_n}$.
		Then, using~(\ref{eq:properties-of-pi}):
		\[
			\pi_i^{\ind\type} \circ \seq{\mmap_1, \dots, \mmap_n}
			=
			\csublr
				{\pi_i^{\ind\type}}
				{\csub
						{\univ_{\ind\type}}
						{\mmap_1, \dots, \mmap_n}}
			= 
				\csub
					{\p{i}{\ind\type}}
					{\mmap_1, \dots, \mmap_n}
			= 
				\mmap_i
		\]
		On the other hand, 
		$
			\seq{ 
				\csub{\pi_1^{\ind\type}}{\mmap},
				\dots, 
				\csub{\pi_n^{\ind\type}}{\mmap}
			}
			= 
				\csubbig 
					{\univ_{\ind\type}}
					{\csub{\pi_1^{\ind\type}}{\mmap},
									\dots, 
									\csub{\pi_n^{\ind\type}}{\mmap}}
			=
				\mmap
		$, 
		so the tensors $\T_{i=1}^n \type_i$ become cartesian products. 
		%
		
		Now suppose that $\clone$ has cartesian structure. By 
			\Cref{res:closure-of-universal-arrows-in-clones}, 
		it suffices to construct a universal arrow for each 
			$\type_1, \dots, \type_n \in {\Obj\clone}^{\star}$.
		To this end, we set 
		\begin{align*}
			\univ_{\ind\type} 
					&:= \seq{\p{1}{\ind\type}, \dots, \p{n}{\ind\type}}
					: \type_1, \dots, \type_n \to \smallprod_{i=1}^n \type_i
			\\
			\ext{\mmap} 
								&:= \csub{\mmap}{\pi_1^{\ind\type}, \dots, \pi_n^{\ind\type}} 
								: \smallprod_{i=1}^n \type_i \to \type[2]
		\end{align*}
		for 
			$\mmap : \type_1, \dots, \type_n \to \type[2]$.
		Then, calculating as in \Cref{res:closure-of-universal-arrows-in-clones}, we have 
		\begin{align*}
			\csub
				{\ext{\mmap}}
				{\univ_{\ind\type}}
			= 
			\csubthree
				{\mmap}
				{\pi_1^{\ind\type}, \dots, \pi_n^{\ind\type}}
				{\seq{\p{1}{\ind\type}, \dots, \p{n}{\ind\type}}}
			=
			\csub
				{\mmap}
				{\p{1}{\ind\type}, \dots, \p{n}{\ind\type}}
			= 
			\mmap
			\\ 
			\ext{ ( \csub{\mmap}{\univ_{\ind\type}} ) }
			= 
			\csubthree
				{\mmap}
				{\univ_{\ind\type}}
				{\pi_1^{\ind\type}, \dots, \pi_n^{\ind\type}} 
			=
			\csubbig
				{\mmap}
				{\dots, 
					\csub
						{\univ_{\ind\type}}
						{\pi_1^{\ind\type}, \dots, \pi_n^{\ind\type}}, 
				\dots}
			= 
			\mmap
		\end{align*}
		so we get a representable structure on $\clone$. 
		\qed
\end{proof}	


\CartCloneEquivalentToCartCat*
\begin{proof}
	The counit of the adjunction 
		$\nucleus{(-)} \dashv \catToClone$
	is the identity.
	On the other hand,  the unit $\eta$ sends a multimap 
		$\mmap : \type_1, \dots, \type_n \to \type[2]$
	to 
		$\csub
			{\mmap}
			{\pi_1^{\ind\type}, \dots, \pi_n^{\ind\type}}$.
	But in 
		\Cref{res:representable-equals-cartesian}
	we showed that this mapping has inverse 
		$\csub
				{(-)}
				{ \seq{\p{1}{\ind\type}, \dots, \p{n}{\ind\type}} } $.
	Hence, $\eta$ is also invertible.   \qed
\end{proof}


%
%


\subsection{ \texorpdfstring{\Cref{sec:closed-clones-as-SK-clones}}{Section 7.1} }

\IsomorphismOfClonesFromIsoOfPresheaves*
\begin{proof}
	Abstractly, it suffices to see that the full subcategory of clones satisfying 
		$\Obj\clone = \sorts$
	is monadic over the functor category $[\sorts^{\star} \times \sorts, \Set]$, and that every monadic functor is an \emph{iso-fibration} in the sense of~\cite{Berger2012}.
	Concretely, the projections are
		$\nu_{\Ctx; \type_i}(\p{i}{\Ctx})$
	and substitution is
	\begin{equation*}
		\label{eq:induced-substitution-operation}
		\csub{\mmap}{\mmap[2]_1, \dots, \mmap[2]_n}
			:=
				\nu_{\Ctx[2]; \type[2]}{\big( 
					\csub
							{ \nu_{\Ctx; \type[2]}^{-1}(\mmap) }
							{\dots, \nu_{\Ctx[2]; \type_i}^{-1}(\mmap[2]_i), \dots} 
				\big)}
	\end{equation*}
	The three axioms are then verified by direct calculation. \qed
\end{proof}


\MainTheoremForClones*
\begin{proof}
	The right-hand equivalence is easily checked. 
	
	For the left-hand one, write $\toCl'$ for the restriction $\toCl \circ \inc$ of $\toCl$ to strict closed clones.
	A short calculation shows that $\toStForCL \circ \toCl' = \id_{\StClClone}$.
	For the other direction, the composite
		$\toCl' \circ \toStForCL$
	sends an extensional SK-clone $\sk\clone$ to the clone 
	with the same objects and hom-sets
	\[	
		(\toCl'\toStForCL\clone)(\Ctx; \type[2])
			= (\toStForCL\clone)(\Ctx; \type[2])
			= \clone(\diamond; \lolliCtx{\Ctx}{\type[2]})
	\]
	Variables and substitution in this clone are as in 
		$\toStForCL\clone$.
	Since the underlying clone structure of 
		$(\toCl'\toStForCL)\clone$
	is the same as that in 
		$\toStForCL\clone$,
	and we have a clone isomorphism
		$\itappu_{\clone} : \toStForCL\clone \to \clone$
	induced by the isomorphisms 
		$\itapp{\Ctx}{\type[2]}$,
	we get a clone isomorphism
		$(\toCl'\toStForCL)\clone \to \clone$.	
	So it remains to show this is an isomorphism of SK-clones.
		
	First we show preservation of application. 
	In any closed clone, $\eval_{\type, \type[2]}$ is the multimap
			$\Lambda(\p{1}{\expObj{\type}{\type[2]}})$	
	so in $\toStForCL\clone$ we get 
		$\eval_{\type, \type[2]} 
			=
			\itapp
				{\lolliObj{\type}{\type[2]}}
				{\lolliObj{\type}{\type[2]}}^{-1}
				(  \p{1}{\lolliObj{\type}{\type[2]}}  )$.
	Now, by extensionality we have a closed multimap 
		$\combI_{\type} : \diamond \to \lolliObj{\type}{\type}$ 
	satisfying
			$\itapp{\type}{\type}(\combI_{\type}) 
					= \clapp{  \combI_{\type}  }{  \p{1}{\type}  } 
					= \p{1}{\type}$.			
	Hence, 
		$\eval_{\type, \type[2]} = \combI_{\lolliObj{\type}{\type[2]}}$.
	We now compute
		 $\clapp{\mmap}{\mmap[2]}$ 
	in
		$(\toCl\toStForCL)\clone$
	as follows:
	\begin{align*}
		{\mmap} \middot_{(\toCl\toStForCL)\clone} {\mmap[2]}
		&= 
		\app{\mmap}{\mmap[2]} 
			\quad \text{ in } \toStForCL\clone			\\
		&=
		\csub
			{\eval_{\type, \type[2]}}
			{\mmap, \mmap[2]} 				
			\quad \text{ in } \toStForCL\clone				\\
		&=
		\csub
			{ \combI_{\lolliObj{\type}{\type[2]}}  }
			{\mmap, \mmap[2]}		
			\quad \text{ in } \toStForCL\clone						\\
		&=
			\itapp{\Ctx}{\type[2]}^{-1}
			\big(\csub
				{ (\itapp
							{\lolliObj{\type}{\type[2]}, \type}
							{\type[2]}
							\,
							\combI_{\lolliObj{\type}{\type[2]}} )}
				{\itapp{\Ctx}
						{\lolliObj{\type}
						{\type[2]}}
						\,
						\mmap, 
					\itapp
						{\Ctx}
						{\type}
						\,
						\mmap[2]}
			\big)
			\quad \text{ in } \clone																	\\
		&=																		
			\itapp{\Ctx}{\type[2]}^{-1}
			\big(\csub
				{ (\clappthree
						{  \combI_{\lolliObj{\type}{\type[2]}} } 
						{ \p{1}{\lolliObj{\type}{\type[2]}, \type}  }
						{\p{2}{\lolliObj{\type}{\type[2]}, \type} }		
				)}
				{\itapp{\Ctx}
						{\lolliObj{\type}
						{\type[2]}}
						\,
						\mmap, 
					\itapp
						{\Ctx}
						{\type}
						\,
						\mmap[2]}
			\big)		
				\quad \text{ in } \clone																		\\
		&=
			\itapp{\Ctx}{\type[2]}^{-1}
			\big(\csub
				{( \clapp 
						{ \p{1}{\lolliObj{\type}{\type[2]}, \type}  }
					{\p{2}{\lolliObj{\type}{\type[2]}, \type} }		
				)}
				{\itapp{\Ctx}
						{\lolliObj{\type}
						{\type[2]}}
						\,
						\mmap, 
					\itapp
						{\Ctx}
						{\type}
						\,
						\mmap[2]}
			\big)	
				\quad \text{ in } \clone	\\
		&=			
			\itapp{\Ctx}{\type[2]}^{-1}
			\big(
				\clapp
					{\csub
							{\p{1}{\lolliObj{\type}{\type[2]}, \type}}
							{\itapp{\Ctx}
													{\lolliObj{\type}
													{\type[2]}}
													\,
													\mmap, 
												\itapp
													{\Ctx}
													{\type}
													\,
													\mmap[2]} 
					}
					{\csub
							{\p{2}{\lolliObj{\type}{\type[2]}, \type}}
							{\itapp{\Ctx}
													{\lolliObj{\type}
													{\type[2]}}
													\,
													\mmap, 
												\itapp
													{\Ctx}
													{\type}
													\,
													\mmap[2]} 
					}
			\big)				
				\quad \text{ in } \clone															\\
			&= 			
			\itapp{\Ctx}{\type[2]}^{-1}
			\big(
				\clapp
					{
						\itapp{\Ctx}
							{\lolliObj{\type}
							{\type[2]}}
							\,
							\mmap					
					}
					{
						\itapp
							{\Ctx}
							{\type}
							\,
							\mmap[2]
					}
			\big)			\quad \text{ in } \clone																	
	\end{align*}
	Hence,
		$\itappu
			({\mmap} \middot_{(\toCl\toStForCL)\clone} {\mmap[2]})
			=	
				{\itappu \, \mmap}
				\middot_{\clone}
				{\itappu \, \mmap[2]}
		$, 
	as required.
		
	For preservation of $\combK$ and $\combS$ first recall that
		$(\itappu_{\clone})_{\diamond; \type[2]}
				= \itapp{\diamond}{\type[2]}$
	is the identity:
		$(\toStForCL\clone)(\diamond; \type[2]) \xra{=} \clone(\diamond; \type[2])$.
	Thus, we need to show that the combinators $\combK$ and $\combS$ in 
		$(\toCl'\toStForCL)\clone$
	are exactly those in $\clone$.
	In $(\toCl'\toStForCL)\clone$ we have that $\combK_{\type, \type[2]}$ is 
				$\Lambda^{\type}\Lambda^{\type[2]} (\p{1}{\type, \type[2]})$
	in $\toStForCL\clone$.
	Using the strictness of $\toStForCL\clone$ and the definition of projections from 
		\Cref{res:isomorphism-of-clones-from-isomorphism-of-presheaves},
	we get that $\combK_{\type, \type[2]}$ in 
		$(\toCl'\toStForCL)\clone$
	is
		$\itapp{\type, \type[2]}{\type}^{-1}
				(\p{1}{\type, \type[2]})$
	in $\clone$.
	But in $\clone$ we have
		$\itapp{\type, \type[2]}{\type} (\combK_{\type, \type[2]})
				= 
				\clappthree
					{\combK_{\type, \type[2]}}
					{\p{1}{\type, \type[2]}}
					{\p{2}{\type, \type[2]}}
				= 
				\p{1}{\type, \type[2]}$
	so $\combK$ must be preserved.
	
	The argument for $\combS$ is similar. \qed
\end{proof}


\subsection{ \texorpdfstring{\Cref{sec:SK-categories}}{Section 8} }

\SKCategoriesAreClosedClones*
\begin{proof}
The main part of the proof is constructing an isomorphism 
\[
	\toSKCat \circ \toCl' : \StClClone \leftrightarrows \StSKCat : \SKCatToStClone
\]
The claim then follows by setting 
	$\SKCatToClone$
to be the composite 
	$\inc \circ \SKCatToStClone$.
The most difficult part of the proof is constructing clone structure from an SK-category
	$\skcat{\cat}$.
Here we give an outline: the rough idea is to define substitution using application so that the $\beta$-law automatically holds.

We begin with some notation. Write 
	$\SS{\type, \type[2]}^{\Ctx}
		: 
		\lolliCtx{\Ctx}{\lolliObj{\type}{\type[2]}}
		\to
		\lolliObj
			{\lolliCtx{\Ctx}{\type}}
			{\lolliCtx{\Ctx}{\type[2]}}
		$
and 
	$\K{\Ctx}{\type}
		: 
		\type
		\to 
		\lolliCtx
			{\Ctx}
			{\type}$
for the maps defined inductively to be: the identity when $\length\Ctx = 0$,
	$\SS{}$
and 
	$\K{}{}$
when $\length\Ctx = 1$, and the composites below for $\length\Ctx \geq 2$:
\begin{align*}
	\SS{\type, \type[2]}^{\Ctx, \objX}
		&:=
		\lolliCtxbig
			{\Ctx}
			{\lolliObj{\objX}{\lolliObj{\type}{\type[2]}}}
		\xra{  \lolliCtx{\Ctx}{\SS{}}  }
		\lolliCtxbig
			{\Ctx}
			{\lolliObj
				{ \lolliObj{\objX}{\type} }
				{ \lolliObj{\objX}{\type[2]} }
			}
		\xra{ \SS{}^{\Ctx} }
			\lolliObjbig
				{\lolliCtx{\Ctx, \objX}{\type}}
				{\lolliCtx{\Ctx, \objX}{\type[2]}}				\\
	\K{\Ctx, \objX}{\type}
		&:= 	
		\type
		\xra{\K{}{}}
		\lolliObj
				{\objX}
				{\type}
		\xra{\K{\Ctx}{}}
		\lolliCtxbig
				{\Ctx}
				{\lolliObj
					{\objX}
					{\type}		
				}				
\end{align*}
The projections $\p{i}{\ind\type}$ are now defined as the image of $\id_{\type_i}$ under the map below:
\[
	\forget\lolliObj{\type_i}{\type_i}
		\xra{ \forget\K{\type_1, \dots, \type_{i-1}}{} }
		\forget\lolliCtx{\type_1, \dots, \type_i}{\type_i}
		\xra{ \forget\lolliCtx{\id}{ \K{\type_{i+1}, \dots, \type_{n}}{} }  }
		\forget\lolliCtx{\type_1, \dots, \type_n}{\type_i} 
\]
Next we define an application operation $\mathsf{app}$ by
\[
	\forget\lolliCtx{\Ctx[2]}{\lolliObj{\type}{\type[2]}}
		\times \forget\lolliCtx{\Ctx[2]}{\type}
		\xra{\forget(\SS{}^{\Ctx[2]}) \times \id}
	\forget\lolliObjbig{ \lolliCtx{\Ctx[2]}{\type[2]} }{ \lolliObj{\Ctx[2]}{\type[2]} }
		\times \forget\lolliCtx{\Ctx[2]}{\type}				
		\xra{\catapp{}}
	\forget\lolliCtx{\Ctx[2]}{\type[2]}
\]
and a unary substitution operation $c$ as follows:
\begin{align*}
	%
	%
	\skforget\lolliObj{\type}{\type[2]}
		\times \skforget\lolliCtx{\Ctx[2]}{\type}
		\xra{  \skforget\K{\Ctx[2]}{} \times \id }
	\skforget\lolliCtx{\Ctx[2]}{\lolliObj{\type}{\type[2]}}
		\times \skforget\lolliCtx{\Ctx[2]}{\type}		
		\xra{\mathsf{app}}
	\skforget\lolliCtx{\Ctx[2]}{\type[2]}	
\end{align*}

We then define substitution by induction. Substitution into a closed term is exactly weakening, so for $\length\Ctx = 0$ we take 
\[
	\skforget\type[2] \times \smallprod_{i=0}^{n} \, \skforget\lolliCtx{\Ctx[2]}{\type_i}
		\xra{\pi_1} 
		\skforget\type[2]
		\xra{\K{\Ctx[2]}{\type[2]}}
		\skforget\lolliCtx{\Ctx[2]}{\type[2]}.
\]

For $\length\Ctx = 1$ we take the map $c$ and for $\length\Ctx \geq 2$ we define $\mathsf{sub}$ as in the next diagram 
	(we abuse notation somewhat by neglecting the structural isomorphisms for the product structure):
\begin{td}
	\forget\lolliCtx{\Ctx, \type_{n+1}}{\type[2]}
			\times \prod_{i=1}^{n+1} \forget\lolliCtx{\type[2]}{\type_i}  & 
	\forget\lolliCtx{\Ctx[2]}{\type[2]} \\
	\forget\lolliCtx{\Ctx}{ \lolliObj{\type_{n+1}}{\type[2]} }
			\times \prod_{i=1}^{n+1} \forget\lolliCtx{\type[2]}{\type_i}	 & 
	\forget\lolliCtx{\Ctx[2]}{ \lolliObj{\type_{n+1}}{\type[2]} }
			\times \forget\lolliCtx{\type[2]}{\type_{n+1}}	
	\arrow["\mathsf{sub}_{\Ctx, \type_{n+1}}", from=1-1, to=1-2]
	\arrow[equals, from=1-1, to=2-1]
	\arrow["\mathsf{sub}_{\Ctx} \times \id"{swap}, from=2-1, to=2-2]
	\arrow["\mathsf{app}"{swap}, from=2-2, to=1-2]
\end{td}

One can show this forms a clone by a long series of inductive arguments using the axioms of \Cref{fig:SK-category-axioms}. A short check shows that it is strict closed. 

Alternatively, and much more simply, one can employ \Cref{res:internal-language-of-SK-cats} and use $\CLwext$ as an internal language: we identify maps 
	$\mmap : \type \to \type[2]$
in $\cat$ with closed terms $\diamond \vdash \mmap : \lolliObj{\type}{\type[2]}$, and take the SK-structure as sketched just before \Cref{res:internal-language-of-SK-cats}.
A short series of inductions yields the following for 
	$\SS{\type, \type[2]}^{\Ctx}$
and
	$\K{\Ctx}{\type}$
as defined above, where $\Ctx := (\varx_i : \type_i)_{i=1, \dots, n}$ and $n \geq 1$:
\begin{align*}
	\appfour
			{\SS{\type, \type[2]}^{\Ctx}}
			{f}
			{g}
			{\varx_1 \, \dots \, \varx_n}
	&=
	\app
		{ (\app{f}{\varx_1 \, \dots \, \varx_n}) } 
		{ (\app{g}{\varx_1 \, \dots \, \varx_n}) } 						\\
	\appthree
			{\K{\Ctx}{\type}}
			{f}
			{\varx_1 \, \dots \, \varx_n}
	&=
		f																					
\end{align*}
It follows that the $i^{\text{th}}$ projection $\p{i}{\ind\type}$ is the closed term $P_i$ which satisfies
$
	\app{P_i}{\varx_1 \, \dots \, \varx_n} = \varx_i
$
and also that, where $\Ctx[2] = (\varx[2]_j : \type[2]_j)_{j=1, \dots, m}$ in the definitions above, then
\begin{align*}
	\app
			{\mathsf{app}(\mmap, \mmap[2])}
			{\varx[2]_1 \, \dots \, \varx[2]_m}
	&=
	\app
		{ (\app{\mmap}{\varx[2]_1 \, \dots \, \varx[2]_m}) } 
		{ (\app{\mmap[2]}{\varx[2]_1 \, \dots \, \varx[2]_m}) } 						\\
	\app
			{c(\mmap, \mmap[2])}
			{\varx[2]_1 \, \dots \, \varx[2]_m}
	&=
		\app
			{\mmap}																				
			{(\app{u}{\varx[2]_1 \, \dots \, \varx[2]_m})}										\\
	\app 
		{\mathsf{sub}{\big( \mmap, (\mmap[2]_1, \dots, \mmap[2]_n) \big)}}
		{\varx[2]_1 \, \dots \, \varx[2]_m}
	&=
		\appfour
			{\mmap}
			{ (\app{\mmap[2]_1}{\varx[2]_1 \, \dots \, \varx[2]_m})  }
			{\dots}
			{ (\app{\mmap[2]_n}{\varx[2]_1 \, \dots \, \varx[2]_m})  }		
\end{align*}

The three clone axioms then reduce to direct calculations with $\CLwext$. Finally, to see this clone is strict closed, one sees as in \Cref{res:main-theorem-for-clones} that 
	$\eval_{\type, \type[2]}$
is the identity $\combI_{\lolliObj{\type}{\type[2]}}$, and hence that for
	$\varx_1 : \type_1, \dots, \varx_n : \type_n \vdash \mmap : \lolliObj{\type}{\type[2]}$ 
the term
	$	\csubbig
			{  \eval_{\type, \type[2]}  }
			{  
				\csub{\mmap}{\p{1}{}, \dots, \p{n}{}},
				\p{n+1}{}
			}$
in the induced clone $\SKCatToStClone\cat$ must satisfy
\begin{align*}
	\appthree
		{\mathsf{sub}
				{\Big(
					\eval_{\type, \type[2]},
					\big( \mathsf{sub}(\mmap, (P_1, \dots, P_n)),  P_{n+1}) \big)
				\Big)}}
		{\varx_1 \, \dots \, \varx_n}
	{a} 																															\\
	&\hspace{-7.5cm}=
	\appthree
		{\eval_{\type, \type[2]}}
		{\Big( 
			\appthree
				{\big( \mathsf{sub}(\mmap, (P_1, \dots, P_n)),  P_{n+1}) \big)}
				{\varx_1 \, \dots \, \varx_n}
				{a}		
		\Big)
		}
		{
			(\appthree 
				{P_{n+1}}
				{\varx_1 \, \dots \, \varx_n}
				{a}	)	
		}																																\\
	&\hspace{-7.5cm}=
	\appthree
		{\eval_{\type, \type[2]}}
		{\Big( 
			\appthree
				{\big( \mathsf{sub}(\mmap, (P_1, \dots, P_n)),  P_{n+1}) \big)}
				{\varx_1 \, \dots \, \varx_n}
				{a}		
		\Big)
		}
		{
			a
		}																																\\		
	&\hspace{-7.5cm}=
	\appthree
		{\eval_{\type, \type[2]}}
		{\big( 
			\appfour
				{\mmap}
				{ 
					(\appthree
						{P_1}
						{\varx_1 \, \dots \, \varx_n}
						{a})
				}
				{\dots}
				{ 
					(\appthree
						{P_n}
						{\varx_1 \, \dots \, \varx_n}
						{a})
				}		
		\big)
		}
		{
			a
		}																																\\			
	&\hspace{-7.5cm}=
	\appthree
		{\eval_{\type, \type[2]}}
		{\big( 
			\appfour
				{\mmap}
				{\varx_1}
				{\dots}
				{\varx_n}		
		\big)
		}
		{
			a
		}																																	\\
	&\hspace{-7.5cm}=
	\appfive
		{\mmap}
		{\varx_1}
		{\dots}
		{\varx_n}	
		{a}			
\end{align*}

Hence, by extensionality, we get that
	$	\csubbig
				{  \eval_{\type, \type[2]}  }
				{  
					\csub{\mmap}{\p{1}{}, \dots, \p{n}{}},
					\p{n+1}{}
				}
		= \mmap
	$
in $\SKCatToStClone\cat$, so this is a strict closed clone.
This extends to a functor
	$\SKCatToStClone : \StSKCat \to \StClClone$
in the obvious manner, and similar calculations---either with the respective internal languages or diagrammatically---show that we get the isomorphism claimed.
\qed
\end{proof}


\CartesianClosedFromSK*
\begin{proof}
	The given isomorphisms are, by induction, equivalent to giving natural isomorphisms of the form
		$\cat({\otimes_i \type_i}, {\type[2]})
				\iso 
				\cat(\tensu, \lolliCtx{\type_1, \dots, \type_n}{\type[2]})$.
	Now, if $\cat$ has an SK-structure together with these isomorphisms then we get
		$\skforget\lolliObj{\otimes_i \type_i}{\type[2]}
			\iso 
			\skforget\lolliCtx{\type_1, \dots, \type_n}{\type[2]}$,
	but this is exactly endowing the closed clone 
		$\SKCatToClone(\cat)$
	with a representable structure, making it cartesian closed. 
	Observing that each of the constructions in 
		\Cref{res:factoring-adjunction-for-cc-clones} 
	and
		\Cref{res:factoring-adjunction-for-closed-clones}
	preserves the underlying category, we see this endows $\cat$ with cartesian closed structure.
	
	Going the other way, if $\cat$ is cartesian closed then so is the clone 
		$\catToClone\cat$, 
	and by representability and \Cref{res:universal-arrows-as-natural-isomorphisms} we get natural isomorphisms 
		$(\catToClone\cat)(\prod_{i=1}^n \type_i; \type[2])
				\iso 
				(\catToClone\cat)(\type_1, \dots, \type_n; \type[2])
		$
	for all $n$. 
	Passing through 
		\Cref{res:factoring-adjunction-for-cc-clones} 
	and
		\Cref{res:factoring-adjunction-for-closed-clones},
	we see that the induced closed category has isomorphisms
	\[
		\cat({\smallprod_{i=1}^n \type_i}, {\type[2]})
		=
		\skforget\lolliObjbig{\smallprod_{i=1}^n \type_i}{\type[2]}
					\iso 
					\skforget\lolliCtx{\type_1, \dots, \type_n}{\type[2]}
		=
		\cat{\big(\terminal, (\expCtx{\type_1, \dots, \type_n}{\type[2]})\big)}
	\]
	as required.
	\qed
\end{proof}